\documentclass[pdflatex,sn-apa, table]{sn-jnl}
\catcode`\|=12\relax 

\usepackage{orcidlink}

\jyear{2022}%

\theoremstyle{thmstyleone}%
\newtheorem{theorem}{Theorem}
\theoremstyle{thmstyletwo}%
\newtheorem{example}{Example}%

\theoremstyle{thmstylethree}%
\newtheorem{definition}{Definition}%

\usepackage[utf8]{inputenc}
\usepackage{graphicx}
\usepackage{algorithmicx}
\usepackage{algpseudocode}
\usepackage[UKenglish]{babel}
\usepackage[UKenglish]{isodate}
\usepackage{amssymb}

\usepackage{amsthm}
\newtheorem{lemma}[theorem]{Lemma}

\newtheoremstyle{case}{}{}{}{}{}{:}{ }{}
\theoremstyle{case}

\def\s#1{\mbox{\boldmath $#1$}}

\algnewcommand\Or{\textbf{or} }

\usepackage[per-mode=fraction, separate-uncertainty = true,multi-part-units=single, tight-spacing=true]{siunitx}
\colorlet{tablebackgroundcol}{gray!12}

\usepackage{amsmath}

\newcommand\norm[1]{\lvert#1\rvert}
\newcommand \romanAlphabet[0]{\Sigma_R}

\algdef{SE}[DOWHILE]{Do}{doWhile}{\algorithmicdo}[1]{\algorithmicwhile\ #1}

\usepackage{tabularx,environ,amsmath,amssymb}
\makeatletter
\newcommand{\problemtitle}[1]{\gdef\@problemtitle{#1}}
\newcommand{\probleminput}[1]{\gdef\@probleminput{#1}}
\newcommand{\problemquestion}[1]{\gdef\@problemquestion{#1}}
\newcommand{\problemoutput}[1]{\gdef\@problemoutput{#1}}%
\NewEnviron{problem}{
  \problemtitle{}\probleminput{}\problemquestion{}
  \BODY
  \par\addvspace{.5\baselineskip}
  \noindent
  \begin{tabularx}{\textwidth}{@{\hspace{\parindent}} l X c}
    \multicolumn{2}{@{\hspace{\parindent}}l}{\@problemtitle} \\
    \textbf{Input:} & \@probleminput \\
    \textbf{Question:} & \@problemquestion
  \end{tabularx}
  \par\addvspace{.5\baselineskip}
}
\NewEnviron{optProblem}{
  \problemtitle{}\probleminput{}\problemoutput{}
  \BODY
  \par\addvspace{.5\baselineskip}
  \noindent
  \begin{tabularx}{\textwidth}{@{\hspace{\parindent}} l X c}
    \multicolumn{2}{@{\hspace{\parindent}}l}{\@problemtitle} \\
    \textbf{Input:} & \@probleminput \\
    \textbf{Output:} & \@problemoutput
  \end{tabularx}
  \par\addvspace{.5\baselineskip}
}
\makeatother

\usepackage{rotating}

\makeatletter
\let\org@insert@column\insert@column
\newcommand*{\dummy@insert@column@h}{%
  \begingroup
    \setbox0=\hbox\bgroup\begingroup
      \org@insert@column
    \endgroup\egroup
    \setbox2=\hbox{}%
    \wd2=\wd0 %
    \ht2=\ht0 %
    \dp2=\dp0 %
    \copy2 %
  \endgroup
}
\newcommand*{\dummy@endpbox}{%
  \@finalstrut\@arstrutbox
  \egroup
  \begingroup
    \setbox0=\lastbox
    \setbox2=\hbox{}%
    \wd2=\wd0 %
    \ht2=\ht0 %
    \dp2=\dp0 %
    \copy2 %
  \endgroup
  \hfil
}
\newcommand*{\dummy@classz}{%
  \@classx
  \@tempcnta \count@
  \prepnext@tok
  \@addtopreamble{%
    \ifcase \@chnum
      \hfil \d@llarbegin
      \dummy@insert@column@h
      \d@llarend \hfil
    \or
      \hskip1sp\d@llarbegin
      \dummy@insert@column@h
      \d@llarend \hfil
    \or
      \hfil\hskip1sp\d@llarbegin
      \dummy@insert@column@h
      \d@llarend
    \or
      $\vcenter
      \@startpbox{\@nextchar}\insert@column \dummy@endpbox
      $%
    \or
      \vtop
      \@startpbox{\@nextchar}\insert@column \dummy@endpbox
    \or
      \vbox
      \@startpbox{\@nextchar}\insert@column \dummy@endpbox
    \fi
  }%
  \prepnext@tok
}
\newcommand*{\tabulardummysetup}{%
  \renewcommand*{\cellcolor}[1]{\null}%
  \let\@classz\dummy@classz
}
\newcommand{\tabularfix}[1]{%
  \def\tabularfix@contents{\ignorespaces#1\ifhmode\unskip\fi}%
  \leavevmode
  \rlap{\tabularfix@contents}%
  \begingroup
    \tabulardummysetup
    \tabularfix@contents
  \endgroup
}
\makeatother

\begin{document}

\title[Heuristics for the RLBWT Alphabet Ordering Problem]{Heuristics for the Run-length Encoded Burrows-Wheeler Transform Alphabet Ordering Problem}


\author*[1]{\fnm{Lily} \sur{Major} \,\orcidlink{0000-0002-5783-8432}}\email{jam86@aber.ac.uk}
\author[1]{\fnm{Amanda} \sur{Clare} \,\orcidlink{0000-0001-8315-3659}}\email{afc@aber.ac.uk}
\equalcont{These authors contributed equally to this work.}

\author[1,2,3]{\fnm{Jacqueline W.} \sur{Daykin} \,\orcidlink{0000-0003-1123-8703}}\email{jwd6@aber.ac.uk}
\equalcont{These authors contributed equally to this work.}

\author[4]{\fnm{Benjamin} \sur{Mora} \,\orcidlink{0000-0002-2945-3519}}\email{b.mora@swansea.ac.uk}
\equalcont{These authors contributed equally to this work.}

\author[1]{\fnm{Christine} \sur{Zarges} \,\orcidlink{0000-0002-2829-4296}}\email{chz8@aber.ac.uk}
\equalcont{These authors contributed equally to this work.}

\affil*[1]{\orgdiv{Department of Computer Science}, \orgname{Aberystwyth University}, \orgaddress{\street{Penglais}, \city{Aberystwyth}, \postcode{SY23 3DB}, \country{UK}}}

\affil[2]{\orgdiv{Department of Information Science}, \orgname{Stellenbosch University}, \orgaddress{\street{Merriman Avenue}, \city{Stellenbosch}, \postcode{7602}, \country{South Africa}}}

\affil[3]{
\orgname{Univ Rouen Normandie, INSA Rouen Normandie, Université Le Havre Normandie,
Normandie Univ}, \orgaddress{\street{LITIS UR 4108}, 
\postcode{F-76000 Rouen}, \country{France}}}

\affil[4]{\orgdiv{Computer Science Department}, \orgname{Swansea University}, \orgaddress{\street{Bay campus, Fabian Way}, \city{Swansea}, \postcode{SA1 8EN}, \country{UK}}}

\abstract{ 

The Burrows-Wheeler Transform (BWT) is a string transformation technique widely used  in areas such as bioinformatics and file compression. 
Many applications combine a run-length encoding (RLE) with the BWT in a way which preserves the ability to query the compressed data efficiently.
However, these methods may not take full advantage of the compressibility of the BWT as they do not 
modify the alphabet ordering for the sorting step embedded in computing the BWT. 
Indeed, any such alteration of the alphabet ordering can have a 
considerable impact on the output of the BWT, in particular on the number of runs.
For an alphabet $\Sigma$ containing $\sigma$ characters, the space of all alphabet orderings is of size
$\sigma!$. While for small alphabets an exhaustive investigation is possible, finding the optimal ordering for larger alphabets is not feasible.
Therefore, there is a need for a more informed search strategy than brute-force sampling the entire space, which motivates a new heuristic approach.
In this paper, we explore the non-trivial cases for the problem of minimizing the size of a run-length encoded BWT (RLBWT) via selecting a new ordering for the alphabet.
We show that random sampling of the space of alphabet orderings usually gives sub-optimal orderings for compression and that a local search strategy can provide a large improvement in relatively few steps.
We also inspect a selection of initial alphabet orderings, including ASCII, letter appearance, and letter frequency. While this alphabet ordering problem is computationally hard we demonstrate gain in compressibility.
}

\keywords{Alphabet ordering, Burrows-Wheeler Transform, Compression, Local search, Random sampling, Run-Length Encoding}

\maketitle

This preprint has not undergone peer review or any post-submission
improvements or corrections. The Version of Record of this article is published in Journal of Heuristics, and
is available online at https://doi.org/10.1007/s10732-025-09548-3

\section{Introduction}

The Burrows-Wheeler Transform (BWT), originally known as block-sorting compression, is a text transformation scheme which computes a permutation of an input string of data \citep{ARTICLE:BWT}. The transformation is achieved by sorting the matrix of all circular shifts (rotations) of a string into lexicographic order and extracting the last column from the matrix. The algorithm can be implemented using a suffix array data structure with overall linear time complexity. Furthermore, using the index of the original string in the sorted matrix, the transform can be inverted in linear time, enabling efficient recovery of the input data. Hence the BWT is applicable to lossless compression activities, notably as a pre-processor preparing the data for compression. An important property of the transform is that it groups together characters with similar context, which are often identical characters, that is, it has a tendency to rearrange a string of characters into runs of the same character.

The computational efficiency of this remarkably simple innovation has enabled wide-ranging applications 
related to the indexing, searching, and compression of text \citep{AdjerohBellMukherjee2008}. The BWT is 
implemented in the popular open-source file compressor Bzip2 \citep{bzip2}, as well as in bioinformatic sequence alignment utilities including Bowtie2 \citep{ARTICLE:bowtie2}, BWA \citep{pmid19451168-BWA}, and SOAP2 \citep{Li2009soap2}, and additionally in image compression \citep{Syahrul2008LosslessImage}.

A large amount of research has been conducted to improve the space complexity of the BWT and common data structures used to query the BWT such as the FM-index.
Such improvements include applying run-length encoding (RLE) to the transformed text or other structures.

This was first investigated by  M{\"a}kinen and Navarro as the run-length encoded FM-index (RLFM-index) \citep{Makinen2005succint}
and improved on \citep{Siren2008runlength}, including much more recently by Gagie, Navarro, and Prezza \citep{gagie2020fully}, using a run-length encoded BWT (RLBWT) and suffix array samples ($r$-index).
The $r$-index provides a full-text searchable index in $O(r)$ space, where $r$ is the number of runs in the BWT,
and has been of interest for the finding maximal exact matches step in bioinformatics read alignment \citep{Rossi2022-da}.

The order of the characters in the compressed BWT text heavily relies on the alphabet ordering used to sort the suffixes. 
By varying the alphabet ordering used for the BWT, the output can be influenced to further group characters, improving over the extended ASCII alphabet ordering which is typically implemented in software utilities.

Alphabet orderings for the BWT have been considered by Chapin and Tate \citep{chapin1998higher}, where both a hand-picked ordering and orderings created from a heuristic algorithm were used and tested on both text and imaging data.
It is suggested that placing similar characters together in the alphabet (vowels, consonants, and punctuation) yield greater compression over ASCII in their pipeline involving the BWT.

Related BWT research has been conducted which, for a given ordered alphabet $\Sigma$, investigate a variety of non-lexicographic orderings of $\Sigma^*$, providing tailored methods to order the set of strings rather than reordering the alphabet.
The ABWT is based on alternating lexicographical order \citep{Giancarlo2018block, giancarlo2020alternating}, which flips the order relation between $<$ and $>$ at each subsequent position during the scan of two strings being compared.
The $V$-BWT is based on $V$-order which repeatedly deletes a $V$-type letter and at the penultimate stage of equality applies co-lexorder \citep{daykin2014vorder}.
The $D$-BWT, which is applied to degenerate (also known as indeterminate) strings where each string position consists of a nonempty subset of letters over $\Sigma$, determines a lex-extension to sort the conjugates \citep{Daykin2017-dbwt}.
The binary $B$-BWT applies binary block order and yields not one but twin transforms \citep{DAYKIN2016118-bBWT}.

Ordering texts within a collection has also been examined \citep{Cox2012largescale, cazaux2019xbwt, ARTICLE:complexityBWTRunsMinimizationAlphabetReordering}, which involves adding a unique separator to each text in the collection. A notable result is that finding a minimal number of runs $r$ by ordering of the texts can be done in linear time \citep{cazaux2019xbwt, ARTICLE:complexityBWTRunsMinimizationAlphabetReordering}, though this approach does not perform reordering of the alphabet within each text.
\cite{pibiri2023weighted} has implemented a graph-based algorithm for ordering the biological sequences within a collection such that the k-mers can be efficiently stored in a hash table. The ordering is chosen to ensure that the corresponding k-mer counts can be compressed using RLE while still enabling queries of the k-mer counts.

Reordering the alphabet to minimize $r$ has been shown to be APX-hard \citep{gibney-thesis,ARTICLE:complexityBWTRunsMinimizationAlphabetReordering}.
In addition, Bentley, Gibney, and Thankachan  showed that finding an alphabet ordering to ensure $r < t$ for a given threshold $t$ is NP-complete \citep{ARTICLE:complexityBWTRunsMinimizationAlphabetReordering}.
We also know that $r$ is limited to no more than twice the number of runs in the original text \citep{mantaci2017measuring}, providing a 
bound to the quality of any worst case alphabet ordering.

While the theoretical hardness of choosing the best alphabet ordering is clear, we still know little about the potential to efficiently and substantially improve on ASCII orderings.
In this paper, we first randomly sample widely from the space and then to use heuristic search to understand more about potential improvements to the alphabet ordering which could be made cheaply and quickly, reducing the size of RLBWT compressed texts.

We show that most randomly sampled alphabet orders achieved
little improvement in size reduction.
However, we show that a First-Improvement local search can quickly improve on randomly sampled alphabet orderings, even when using a limited number of steps. 
Additionally we consider initializing the search with promising initial orderings to seed the search, taking into account character frequency and appearance, and showing how they affect the speed of improvement of the search over time. 
We also evaluate a variety of operators for directing the search and find that a combined search strategy of two standard operators \textsc{Swap} and \textsc{Insert} \citep{EibenSmith2015} may improve the local minima depending on the order of their use.
We further consider varied orderings of the neighbors of an alphabet ordering, searching them lexicographically, reverse-lexicographically, and randomly (Section \ref{sec:method-ls}).

Section~\ref{sec:problem} introduces and formalizes the problem and terminology. 
We discuss the special cases for small alphabets in Section \ref{sec:small}, and First-Improvement local search methods for larger alphabets in Section \ref{sec:larger}.
Our experimental setup is discussed in Section~\ref{sec:methods-setup}. A detailed discussion of our results can be found in Section~\ref{sec:results}. We summarize our contribution in Section~\ref{sec:conclusion} and provide an overview of future research directions.

\section{Notation, Problem Definition, and Modeling}
\label{sec:problem}

In this section, we provide formal definitions for BWT and RLE with examples given for the key definitions. We also formally state the considered optimization problem.
For each of the following definitions it is assumed that we have the following:

An \emph{alphabet} $\Sigma$, is an ordered
 non-empty set of unique characters $\{x_0, x_1, \dots, x_{\sigma-1}\}$, where $x_0 < x_1 < \cdots < x_{\sigma-1}$ and $x_i < x_j$ implies $x_i$ precedes $x_j$.
The standard `Roman' alphabet ordering as implemented in the ASCII table is denoted as $\romanAlphabet$.

A \emph{string} \s{s} = $c_0 c_1 \dots c_{n-1} = \s{s}[0..n-1]$ over an \emph{alphabet} $\Sigma$, is a finite sequence of characters of length $n=\norm{\s{s}}$,  such that $c_i \in \Sigma, c_i = \s{s}[i]$. The individual characters $c_i$ will be denoted as $\s{s}[i]$ subsequently.

\subsection{Key Definitions}

A BWT (Definition~\ref{def:bwt}) is computed by first forming a Burrows-Wheeler Matrix (BWM, Definition~\ref{def:bwm}) -- the sorted list of strings formed by all \emph{cyclic rotations} of a string -- and taking the last column of the matrix.
This has a tendency to group identical characters together and is useful for compressing with an RLE.

Alphabet reordering allows us to impact the ordering of the rows within a BWM.
For any \emph{string} \s{s} of length $n$ with an \emph{alphabet} $\Sigma$ of size $\sigma$, there is a total of $\sigma!$ possible alphabet orderings.

\begin{definition}(Burrows-Wheeler Matrix / BWM)
    \label{def:bwm}    
    Let \s{s} and \s{s'} be two strings over the same alphabet $\Sigma$. The string \s{s'} is said to be a \emph{cyclic rotation} of \s{s} if and only if there exists two strings \s{u} and \s{v} with $\norm{\s{v}} = 1$ such that $\s{s}=\s{u}\s{v}$ and $\s{s'}=\s{v}\s{u}$.
    The lexicographically ordered, row-arranged set of all \emph{cyclic rotations} of $\s{s}\$$ sorted according to the order of $\Sigma$, with \$ least in $\Sigma$, and denoted as $BWM(\s{s}, \Sigma)$ is the BWM of $\s{s}$.
    The BWM of \s{s} without adding an implicit \$ is denoted as $BWM_*(\s{s}, \Sigma)$.
\end{definition}

It should be noted that we do not consider the end marker (normally \$) to be movable within the alphabet ordering, instead always appearing first as an implicit least character.
We illustrate Definition~\ref{def:bwm} in the following example. For the sake of simplicity we use the standard `Roman' alphabet ordering $\romanAlphabet$.

\begin{example}
\label{example:bwm-roman}
    Let a string $\s{s} = cacatcg$ and using $\romanAlphabet$, the BWM is as follows, where $F$ and $L$ denote the first and last columns of the matrix respectively:
\[
	BWM(\s{s}, \romanAlphabet) = \begin{vmatrix}
		F & & & & & & & L\\
		\$&c&a&c&a&t&c&g\\
		a&c&a&t&c&g&\$&c\\
		a&t&c&g&\$&c&a&c\\
		c&a&c&a&t&c&g&\$\\
		c&a&t&c&g&\$&c&a\\
		c&g&\$&c&a&c&a&t\\
		g&\$&c&a&c&a&t&c\\
		t&c&g&\$&c&a&c&a\\
	
		\end{vmatrix}
\]
\end{example}

Using Definition~\ref{def:bwm} we can now formally define a BWT as the last column $L$ of the BWM.

\begin{definition}(Burrows-Wheeler Transform / BWT)
\label{def:bwt}

    For an input string \s{s} of length $n$ and an alphabet ordering $\Sigma$, we define the BWT as BWT(\s{s}, $\Sigma$) = BWM(\s{s}, $\Sigma$)[$i$, $n-1$], $\forall i$ such that $ 0 < i < n-1$.
    The BWT of \s{s} without implicitly adding a \$ is denoted as BWT$_*(\s{s}, \Sigma)$.
\end{definition}

The BWT for the string in Example \ref{example:bwm-roman} with the ordering $\romanAlphabet$ is therefore obtained from the last column $L$ of the BWM as $gcc\$atca$. The BWT can be computed from the suffix array (SA) of \s{s} as follows:
\begin{equation}
\nonumber
  L[i]=\begin{cases}
    \$, & \text{if SA[i]=0}\\
    $\s{s}[SA[i]-1]$, & \text{otherwise}.
  \end{cases}
\end{equation}

The suffix array may be computed in linear time \citep{KO2005143, KIM2005126} 
and enables BWT construction with the same time complexity. Inversion of the BWT may be achieved through the last-first mapping \citep{ARTICLE:BWT}.
For a string \s{s}, the BWT(\s{s}, $\Sigma$) may be updated for a modified string \s{s'}, BWT(\s{s'}, $\Sigma$) in sub-linear time \citep{10.1145/3519935.3520061}.
Recent developments have led to an efficient algorithm for directly constructing the RLBWT \citep{nishimoto2022rlbwt} in
$O(n + r \log r)$ time but $O(n)$ time for strings where $r = O(n/\log n)$, and working space $O(r\log n)$ of bits.

\begin{definition}(Run-Length Encoding / RLE)
    \label{def:rle}
    For a string \s{s}, \s{u} is a \emph{substring} of \s{s} if there are possibly empty strings \s{a} and \s{b} such that \s{s} = \s{a}\s{u}\s{b}, and \s{u} is bounded by the indexes ($i$, $j$) where $0 \le i \le j \le n-1$.
    Let $\s{p_0} \s{p_1} \s{p_2} \dots \s{p_{n-1}}$ be a sequence of \emph{substrings} of \s{s} such that \s{s} = $\s{p_0} \s{p_1} \s{p_2} \dots \s{p_{n-1}}$ and $\s{p_i}[j]=\s{p_i}[k]$ $\forall j,k<\norm{\s{p_i}}$ (i.e., each $\s{p_i}$ is a sequence containing identical characters).
    RLE(\s{s}) is the string where each run of identical characters $\s{p_i}$ in \s{s} is replaced by a single copy of that character and a count of the characters in the run.
    We denote the length of each run of characters in superscript, for instance: RLE(\s{s}) = ${\s{p_0}[0]}^{\norm{\s{p_0}}} {\s{p_1}[0]}^{\norm{\s{p_1}}} \ldots {\s{p_{n-1}}[0]}^{\norm{\s{p_{n-1}}}}$.
\end{definition}

We illustrate Definition~\ref{def:rle} with an example using $\romanAlphabet$.
\begin{example}
    Let a string $\s{s} = cacatcg$ using $\romanAlphabet$,
    The output $L$ column of $BWM(\s{s}, \romanAlphabet)$ is $gcc\$atca$. This is encoded as $\textrm{RLE}(BWT(\s{s}, \Sigma)) = g^1c^2\$^1a^1t^1c^1a^1$, $\norm{\textrm{RLE}(BWT(\s{s}, \Sigma))} = 14$.
\end{example}

We are interested in overall memory usage for the RLE, so each run is encoded as the byte (character) and then the length of the run up to 255. Any runs over 255 in length are encoded as additional bytes (Section~\ref{sec:problem-statement}). This is represented as $\norm{RLE(\s{s})}$.

To demonstrate the influence an alphabet ordering can have on the BWM, BWT and RLE, we consider the same example string using an alternative alphabet ordering.

\begin{example}
    Let a string $\s{s} = cacatcg$, $\Sigma = \$ < a < g < c < t$.
\[
	BWM(\s{s}, \Sigma) = \begin{vmatrix}
		F & & & & & & & L\\
		\$&c&a&c&a&t&c&g\\
		a&c&a&t&c&g&\$&c\\
		a&t&c&g&\$&c&a&c\\
	    g&\$&c&a&c&a&t&c\\
        c&a&c&a&t&c&g&\$\\
		c&a&t&c&g&\$&c&a\\
		c&g&\$&c&a&c&a&t\\
		t&c&g&\$&c&a&c&a\\
		\end{vmatrix}
\]
With this new ordering $\Sigma$, $\textrm{RLE}(BWT(\s{s}, \Sigma)) = g^1c^3\$^1a^1t^1a^1$, $\norm{\textrm{RLE}(BWT(\s{s}, \romanAlphabet))} = 12$.
\end{example}

\begin{definition}($r$)
    The number of maximal length runs in the BWT.
    For example, let $\s{s} = acacacbbacbac$ and $\Sigma = a < b < c$.
    BWT$_*(\s{s}, \Sigma) = bccbccbcaaaaa$. As there are 7 runs in the BWT output, $r = 7$.
\end{definition}

The $r$ value for any BWT of a string is not necessarily less than the $r$ value of a string itself. A classic example is the string $mississippi\$$ where the BWT is $ipssm\$pissii$ which also shows that computing the BWT permutation of the input data does not necessarily make the data more compressible.

\subsection{Problem Statement}
\label{sec:problem-statement}

We wish to investigate the sample space of alphabet orderings. Since reordering the alphabet to minimize the number of runs has been shown to be APX-hard, we will use and investigate different heuristics.

We evaluate the fitness of any new alphabet ordering $\Sigma$ by the total length of its RLE.
As we use the input data bytewise, the `characters' in our alphabet are these bytes and not another encoding such as UTF.
We do not use the total number of identical character runs in the BWT($r$) as the fitness since we are interested in overall memory consumption of the representation.
Each input file is read bytewise, so multi-byte (non-ASCII) characters are represented as more than one `character' in the alphabet.

We encode the RLE as a sequence of byte pairs, with the first byte of each pair representing the run's character, and the second byte representing the length of the run.
As the maximum value that may be represented in a byte is 255, any larger runs will be represented with multiple pairs of bytes.
We therefore seek to minimize the size of RLE for our tested texts.

\begin{optProblem}
  \problemtitle{\textsc{RLBWT Alphabet Ordering Problem}}
  \probleminput{A string \s{s} of length $n$ over an alphabet $\Sigma$ of size $\sigma$}
  \problemoutput{A reordering of $\Sigma$ for $RLE(BWT(\s{s}, \Sigma))$ that minimizes $\norm{RLE(BWT(\s{s}, \Sigma))}$.}
\end{optProblem}

\section{Methods for Small Alphabets}
\label{sec:small}

\subsection{Binary Alphabet Orderings}

We consider the simplest non-trivial case where there are only 2 characters and show that when using the reverse of an alphabet ordering that $r$ remains the same for both orderings for primitive strings.

\begin{lemma}
Let \s{s} be a binary string over $\Sigma = \{a, b\}$ of length $n$.
Let $\Sigma_1 = a < b$, $\Sigma_2 = b < a$.
Then $r(\s{s}, \Sigma_1) = r(\s{s}, \Sigma_2)$.
\end{lemma}

\begin{proof}
Suppose \s{s} is primitive, then all of its conjugates are distinct \citep{PETERSEN1996141-propertiesofprimitivewords}. Let \s{p} = $p_0 \dots p_{n-1}$ and \s{q} = $q_0 \dots q_{n-1}$ be two adjacent rows in BWM$_*(\s{s}, \Sigma_1)$ such that \s{p} is lexicographically less than \s{q}. Assume that \s{s} contains two distinct characters (otherwise the claim is trivial), then let $t$ be the minimal index such that $p_t \neq q_t$, thus $p_t = a < b = q_t$. On the other hand in BWM$_*(\s{s}, \Sigma_2)$, \s{q} is lexicographically less than \s{p}. The argument holds for all pairs of adjacent rows showing that the two matrices are flipped row-wise, hence have the same $r$ value.\\
\indent In the case that \s{s} is not primitive and has the form $\s{u}^k$, $k>1$, then all groups of $k$ identical and adjacent rows in BWM$_*(\s{s}, \Sigma_1)$ will likewise be adjacent in BWM$_*(\s{s}, \Sigma_2)$ after the flipping.
\end{proof}

Observe that using transitivity the argument on reversing an alphabet extends to an arbitrary finite alphabet which motivates our search for effective orderings. We illustrate concepts with the following ternary example:

\begin{example}(Ternary primitive)

Let \s{s} = $aabbcc$, and $\Sigma_1 = a < b < c$:
\[
	BWM_*(\s{s}, \Sigma_1) = \begin{vmatrix}
		F & & & & & L\\
		a & a & b & b & c & c\\
        a & b & b & c & c & a\\
        b & b & c & c & a & a\\
        b & c & c & a & a & b\\
        c & a & a & b & b & c\\
        c & c & a & a & b & b\\
		\end{vmatrix}
\]
BWT$_*(\s{s}, \Sigma_1)$ = $caabcb$, $r(\s{s}, \Sigma_1)$ = 5.

In the case with $\Sigma_2 = c < b < a$:
\[
	BWM_*(\s{s}, \Sigma_2) = \begin{vmatrix}
		F & & & & & L\\
	    c & c & a & a & b & b\\   
        c & a & a & b & b & c\\
        b & c & c & a & a & b\\
        b & b & c & c & a & a\\        
        a & b & b & c & c & a\\
        a & a & b & b & c & c\\
		\end{vmatrix}
\]
BWT$_*(\s{s}, \Sigma_2)$ = $bcbaac$, $r(\s{s}, \Sigma_1)$ = 5.

Thus both orderings have the same $r$ value.
\end{example}

\begin{example}(Ternary non-primitive)

    Let \s{s} = abcabc, and $\Sigma_1 = a < b < c$:
    \[
	BWM_*(\s{s}, \Sigma_1) = \begin{vmatrix}
		F & & & & & L\\
	    a&b&c&a&b&c\\
        a&b&c&a&b&c\\
        b&c&a&b&c&a\\
        b&c&a&b&c&a\\
        c&a&b&c&a&b\\
        c&a&b&c&a&b\\
		\end{vmatrix}
    \]

    Likewise for $\Sigma_2 = c < b < a$ it is trivial to see that the $r$ value of both orderings is the same.

    \[
	BWM_*(\s{s}, \Sigma_2) = \begin{vmatrix}
		F & & & & & L\\
        c&a&b&c&a&b\\
        c&a&b&c&a&b\\
        b&c&a&b&c&a\\
        b&c&a&b&c&a\\
        a&b&c&a&b&c\\
        a&b&c&a&b&c\\
		\end{vmatrix}
    \]
\end{example}

\subsection{Exhaustive Search on Biological Data}
\label{sec:results-exhastive-search}

For small alphabets it may be feasible to search through all possible alphabet orderings to find the one(s) that  provide(s) the best RLBWT compression of the data.

For example, in the case of genomic data, such as a collection of the genome sequences of many \textit{E. coli} bacteria, we could expect a limited 4-letter alphabet \{$a$,$c$,$g$,$t$\}  representing nucleotides  giving 24 possible alphabet orderings. 
For such a collection, the genomes would share much in common, and the RLBWT of a concatenation of these sequences should capture the commonalities effectively. 
However, almost all bioinformatics algorithms choose to use the ASCII ordering $a < c < g < t$, even though this may not provide the best results.

We took a collection of 150 diverse \textit{E. coli} genomes from NCBI and compared the compression obtained using each possible alphabet ordering. 
The ordering $t < c < a < g$ provided the best compression (-91.627\% change), and the worst ordering was $c < g < t < a$ (-91.576\% change) -- for further details see Table \ref{tab:150-ecoli-stats}. 
Although an exhaustive search through all possible alphabet orderings is often prohibitively expensive, this example demonstrated that better choices can make improvements, and motivated our investigation to find such orderings and examine the search space of orderings for its properties.

\begin{table}[h]
\centering
\rowcolors{1}{}{tablebackgroundcol}
\caption{Percentage difference in file size using each of the 24 alphabet orderings for the alphabet \{a,c,g,t\} for the 150 concatenated \textit{E.coli} data files when using the RLBWT.}
\label{tab:150-ecoli-stats}
\begin{tabular}{|S[detect-weight, table-format=-1.3]|S[detect-weight, table-format=-1.3]|S[detect-weight, table-format=-1.3]|S[detect-weight, table-format=1.3]|S[detect-weight, table-format=-1.3]|S[detect-weight, table-format=-1.3]|S[detect-weight, table-format=1.3]|}
\hline
{Min \% Change} & {Max \% Change}  & {Mean} & {Std}\\
\hline
-91.576&-91.627&-91.597&0.02\\
\hline
\end{tabular}%
\end{table}

\section{Methods for Larger Alphabets}
\label{sec:larger}

Since the permutation space of the alphabet ordering for an alphabet of size $\sigma$ is $\sigma!$, in most practical cases, exhaustive enumeration of all alphabet orderings is not a feasible approach. We therefore consider different variants of a First-Improvement local search for larger alphabets. We use different types of texts, considering a variety of different text lengths and alphabet sizes. The main goal of our experimental analysis is to provide insights into the working principles of the considered methods for the given problem and to provide guidelines for their use. The considered algorithms are introduced in Sections~\ref{sec:method-sampling} and \ref{sec:method-ls}.

\subsection{Baseline: Random Sampling}
\label{sec:method-sampling}

We use uniform random sampling as a baseline approach to inspect the statistical distribution of potential compression gains that could be made by reordering the alphabet. The results are compared with the results of First-Improvement local search as described in the next section.

An interesting property of random sampling is that the mean number of improvements to be expected is actually bounded by the logarithm of number of trials. Indeed, as every new sample (i.e., alphabet order) is independent of the previous samples, the chance of obtaining a better compression after $T$ sampling steps with a new sample would be $P_{improving(T)}=\frac{1}{T+1}$, should the compression value obtained be unique for every sampling event. Therefore, the expected number of successive improvements in compression using $T$ random samples would be simply given by $\sum_{i=0}^{T} \frac{1}{i+1} = O(\log{T})$. In practice, the compression value will belong to a limited set of integer values, and each new sample result may just be equal to the best compression value obtained so far. As such, the actual chance $P_{improving}(T)$ after $T$ samples to improve compression is less than $\frac{1}{T+1}$, and $O(\log{T})$ becomes an upper bound for the total number of improvements expected from $T$ random samples.

\subsection{First-Improvement Local Search}
\label{sec:method-ls}

We consider a variant of First-Improvement Local Search as our main optimisation approach. Pseudocode for this approach is given in Algorithm~\ref{algo:local-search}. The algorithm takes as input a text and its alphabet. It starts from some initial alphabet ordering $\pi$ (line 3) and tries to improve the ordering until a provable local optimum is reached (line 15). In each loop it considers all neighbors of the current ordering $\pi$ in a given order (line 7). If an improvement is found (line 8), the algorithm moves to the better ordering (line 10) and repeats the process (updating the neighborhood as needed, line 11). The fittest ordering is returned.

In our experiments, we consider 9 different initialization methods in line 3 of Algorithm~\ref{algo:local-search} as discussed in Section~\ref{sec:method-ls-init}. We also consider 12 different neighborhoods as discussed in Section~\ref{sec:method-ls-neighbors}, leading to a total number of 108 algorithm configurations.

\begin{algorithm}
\caption{First-Improvement Local Search}
\label{algo:local-search}
\begin{algorithmic}[1]
\State \textbf{Input:} A string $\s{w}$ over an alphabet $\Sigma$ of size $\sigma$
\State \textbf{Output:} An ordering $\pi$ of $\Sigma$
\State $\pi \leftarrow \textsc{initialiseOrder}(\Sigma)$ \Comment{Initialisation (Section~\ref{sec:method-ls-init})}\;
\State $N \leftarrow \textsc{initialiseNeighborhood}(\pi)$ \Comment{Neighbors of $\pi$ (Section~\ref{sec:method-ls-neighbors})}\;
\Do
    \State improvementFound = false\;
    \For{all $\pi'\in N$}
        \If{$f(\pi') < f(\pi)$}\Comment{Neighbor is improvement}
            \State improvementFound = true\;
            \State $\pi \leftarrow \pi'$\;
            \State $N \leftarrow \textsc{updateNeighborhood}(\pi)$\;
            \State \textbf{break}\;
        \EndIf
    \EndFor
\doWhile{(improvementFound)}
\State \Return $\pi$\;
\end{algorithmic}
\end{algorithm}

\subsubsection{Initialization}
\label{sec:method-ls-init}

We consider the following 9 initialization methods in line 3 of Algorithm~\ref{algo:local-search}. With the exception of random initialization, all methods are deterministic and based on some heuristic or standard ordering from the literature. 

\begin{itemize}
    \item \textbf{Random:} We determine 20 random orderings using Fisher-Yates shuffle~\citep{alma99376173402418-fisher-yates}. The same 20 orderings are used for all experiments using random initialization.
    \item \textbf{ASCII:} The (extended) ASCII ordering of the alphabet.
    \item \textbf{First Appearance:} Characters are ordered by their order of appearance in the text.
    \item \textbf{Least Frequent:} Characters are ordered by the number of occurrences in the text, least frequent first.
    \item \textbf{Most Frequent:} Characters are ordered by the number of occurrences in the text, most frequent first.
    \item \textbf{Chapin-Tate:} A hand-tuned ordering used by Chapin and Tate \citep{chapin1998higher}\footnote{And personal communications with B. Chapin} for a similar compression problem based around the BWT.

    The ASCII alphabet ordering, however `@' replaces `!', `+,-.' is rearranged to `+-,.', AEIOU and aeiou are both brought to the front of each block of both upper and lower case. The consonants are reorganised as $B < C < D < G < F < H < R < L < S < M < N < P < Q < J < K < T < W < V < X < Y < Z.$

   


    \item \textbf{Inverse Permutation Chapin-Tate:}  The inverse permutation \citep{knuthArt} of the Chapin-Tate ordering, spacing out the vowels and reordering the consonants. Since grouping vowels is important for the Chapin-Tate ordering we investigate the effects of not doing this.

    The vowels are ordered as $A < I < O < U < E$, however they are interspersed through the ordering instead of being together at the start of the ordering.
    This results in an ordering of $A < F < G < H < B < J < I < K < C < S < T < M < O < P < D < Q < R < L < N < U < E < W < V < X < Y < Z$.
    Other changes to the ordering such as `!' and `@' being swapped, and the `+-,.' rearrangement are also present in this ordering.
    
    \item \textbf{Vowels}: Vowels (aeiouAEIOU) are placed at the beginning of the ordering. Similar to the Chapin-Tate ordering, the main aim is to explore if grouped punctuation and the consonant reordering really helps the problem, or if moving the vowels alone will yield a better initial ordering.
    
    \item \textbf{FDA:} The FDA algorithm was introduced to determine an alphabet ordering for a variant of the Lyndon factorization problem~\citep{10.1007/978-3-030-58112-1_27}. The main motivation to include this method is to explore if these orderings might be more generally useful for alphabet ordering problems on strings.
    
    It should be noted that FDA determines a partial alphabet ordering. A total ordering is produced from the partial ordering by topological sort.
\end{itemize}

For each of these orderings, we place the selected end marker character (usually \$) as least in the ordering when performing the BWT.

\subsubsection{Local Search Neighborhoods}
\label{sec:method-ls-neighbors}

For the neighborhoods in lines 4 and 11 of Algorithm~\ref{algo:local-search} we consider combinations of two standard operators for permutation sample spaces, namely \textsc{Swap} (aka \textsc{Exchange}) and \textsc{Insert} (aka \textsc{Jump})~\citep{EibenSmith2015}. \textsc{Swap} picks two integers $0 \leq i < j \leq \sigma-1$ and swaps the characters at positions $i$ and $j$. \textsc{Insert} picks two integers $0 \leq i, j \leq \sigma-1$ with $i \neq j$. It moves the character at position $i$ to position $j$, shifting all subsequent characters to the right.

Looking at the two operators in isolation we observe that \textsc{Swap} yields a neighborhood of size $\sigma(\sigma-1)/2$ while \textsc{Insert} yields a neighborhood of size $\sigma(\sigma-1)$. Both neighborhood sizes are quadratic in $\sigma$, the size of the alphabet.

We first consider both operators in isolation and investigate three different orderings of the neighbors in the neighborhoods:
\begin{itemize}
    \item \textbf{Random Order:} Using a random order of the neighbors is the most common approach.
    \item \textbf{Lexicographic Order (\textsc{Lex}):} We hypothesize that it maybe be beneficial to first fix characters at the start of the ordering. We therefore consider the fixed lexicographic order of neighbors. For example, for \textsc{Swap} we consider $i$-$j$-pairs in the following order: $(0,1), (0, 2), ..., (0, \sigma-1), (1, 2), ..., (\sigma-2, \sigma-1)$.
    Any unspecified neighborhood ordering should be assumed to be \textsc{Lex}.
    \item \textbf{Reverse Lexicographic Order (\textsc{RevLex}):} We consider the opposite case by reversing the lexicographic order given above.
\end{itemize}

Finally, we consider a combination of \textsc{Swap} and \textsc{Insert}. More precisely, we first try all possible \textsc{Swap}s followed by all possible \textsc{Insert}s and vice versa. For each of the two lists of neighbors we consider all three orders defined above, ensuring that all \textsc{Swap}s are sorted before all \textsc{Insert}s (and vice versa) as appropriate.
If a combination of operators is used, the second operator will only be used until an improvement is found. The algorithm then returns to using the first operator.

\section{Results and Discussion}
\label{sec:results}

\subsection{Experimental Setup}
\label{sec:methods-setup}
All our experiments are run on Super Computing Wales\footnote{\url{https://portal.supercomputing.wales/index.php/about-sunbird/}}
on a single node (2x Intel(R) Xeon(R) Gold 6148 CPU @ 2.40GHz with 20 cores each) and Aberystwyth DCS cluster on a single node (2x Intel(R) Xeon(R) Gold 6248R CPU @ 3.00GHz with 48 cores each).
As usual, we report the number of objective function evaluations rather than wall-clock time.

For each stochastic variant of First-Improvement local search, we report statistics on the results of 20 independent runs. However, to avoid the distortion of our results due to different random starting points, all variants with random initialization use the same 20 starting points which were randomly determined prior to running our experiments. This way we can analyze the effect of different neighborhoods on the same starting points without introducing additional variables. 
For each run we report the number of function evaluations as `steps'.

We define $C$ as the percentage change in file size relative to the uncompressed size (measured as a percentage in bytes):
\begin{equation}
\label{eq:percent-change}
C = \left(\dfrac{\text{Compressed Size} - \text{Uncompressed Size}}{\text{Uncompressed Size}}\right) \cdot 100
\end{equation}
We present $C$ as raincloud plots \citep{10.12688/wellcomeopenres.15191.2} (a combination of a distribution, boxplot, and jittered point cloud) to give an indication of the density and shape of the sample space. A negative value for $C$ demonstrates a reduction in size while a positive value demonstrates an increase in size. The smaller the value for $C$ the better the compression. The number of steps presented is hitting time and not exhaustive checking of the neighbors.

\subsection{Benchmarking}
\label{sec:results-corpus-description}

We use a standard benchmark for data compression for our analysis, namely the Canterbury corpus \citep{582019}. Table~\ref{tab:cantbry-corpus-files} lists the different files contained in the corpus, including the size of each file and the corresponding alphabet size.
The file \texttt{kennedy.xls} is excluded from our experiments for technical reasons: our BWT implementation relies on a unique end marker character and the alphabet size of 256 leaves no available character if the file is used bytewise. 
We reorder the alphabet by mapping the input text characters to new ones based on the order of characters in the alphabet. We then use the SAIS suffix array implementation by Yuta Mori\footnote{\url{https://sites.google.com/site/yuta256/sais}} to compute the suffix array \citep{5582081, KO2005143}. 
Another implementation is also provided in our repository but was not used due to speed.

We do not run to completion for some files in the corpus (for example: ptt5, sum, xargs.1) for some of the \textsc{Swap then Insert} and \textsc{Insert then Swap} methods due to prohibitive runtimes. Instead we run to a limit of 10,000 steps as this is more than the maximum number of steps to outperform a random sample (Table.~\ref{tab:local-search-steps-to-beat-random}). A full list of the files and methods run until 10,000 steps only is available in our repository\footnote{\url{https://github.com/jam86/Heuristics-for-the-Run-length-Encoded-Burrows-Wheeler-Transform-Alphabet-Ordering-Problem}}.

\begin{table}[tb]
    \centering
    \rowcolors{1}{}{tablebackgroundcol}
    \tabularfix{
    \begin{tabular}{|l|l|S[table-format=7.0]|l|}
        \hline
        File & Description & {Bytes} & Alphabet \\
        \hline
        \texttt{alice29.txt} & The text of Alice's Adventures in Wonderland & 152089 & 74\\
        \texttt{asyoulik.txt} & Text from Shakespeare's play As You Like It & 125179 & 68\\
        \texttt{cp.html} & HTML with a large number of links & 24603 & 86\\
        \texttt{fields.c} & C source code & 11150 & 90\\
        \texttt{grammar.lsp} & LISP source code& 3721 & 76\\
        \texttt{kennedy.xls} & Microsoft Excel document & 1029744 & 256\\
        \texttt{lcet10.txt} & Conference Proceedings & 426754 & 84\\
        \texttt{plrabn12.txt} & Text from John Milton's Paradise Lost & 481861 & 81\\
        \texttt{ptt5} & Fax data & 513216 & 159\\
        \texttt{sum} & Sun SPARC executable & 38240 & 255\\
        \texttt{xargs.1} & GNU Man page for xargs & 4227 & 74\\
        \hline
    \end{tabular}
    }
    \caption{Files in the Canterbury corpus with their size in bytes and the number of unique bytes in their alphabet}
    \label{tab:cantbry-corpus-files}
\end{table}

\subsection{Randomly Sampled Alphabet Orderings}
\label{sec:results-random-sampling}

We inspected the landscape of percentage change in compression that can be achieved using the RLBWT by changing the alphabet order, and sampled 240,000 alphabet orders uniformly at random (by Fisher-Yates shuffle) for each of the texts. 
We do this to learn more about the shape of the sample space and to understand whether there are many best orderings to be found or few.

The number of samples was chosen since it was large but remained tractable to compute on a local machine (Intel(R) i7-8700K CPU @ 3.70GHz with 12 cores). 
To exemplify our findings, the results for files \texttt{alice29.txt},  \texttt{sum} and \texttt{fields.c} are shown in 
Fig.~\ref{fig:random-sample-cantrbry}.
The number of samples may be few in comparison with the very large $\sigma!$ space, however due to the smooth  overall distributions without  outliers  we can see that the sampling already gives a clear picture of the shape of the space from which further random samples would be obtained. 
The distributions are bell-shaped but not normally distributed (scistats.normaltest \cite{2020SciPy-NMeth}), having a long thin tail downwards where better solutions can be found.

These distributions demonstrate a spread of percentage compression for different alphabet orders, with the majority being sub-optimal choices and the sample space having only a thin tail of better choices.
These figures also highlight the surprisingly good compression achieved by the ASCII ordering, lying far below most random choices of alphabet order, even for executable files such as \texttt{sum}. The full set of figures for all corpus files can be seen in our repository.

\begin{figure}
    \centering
        \includegraphics[width=0.6\textwidth]{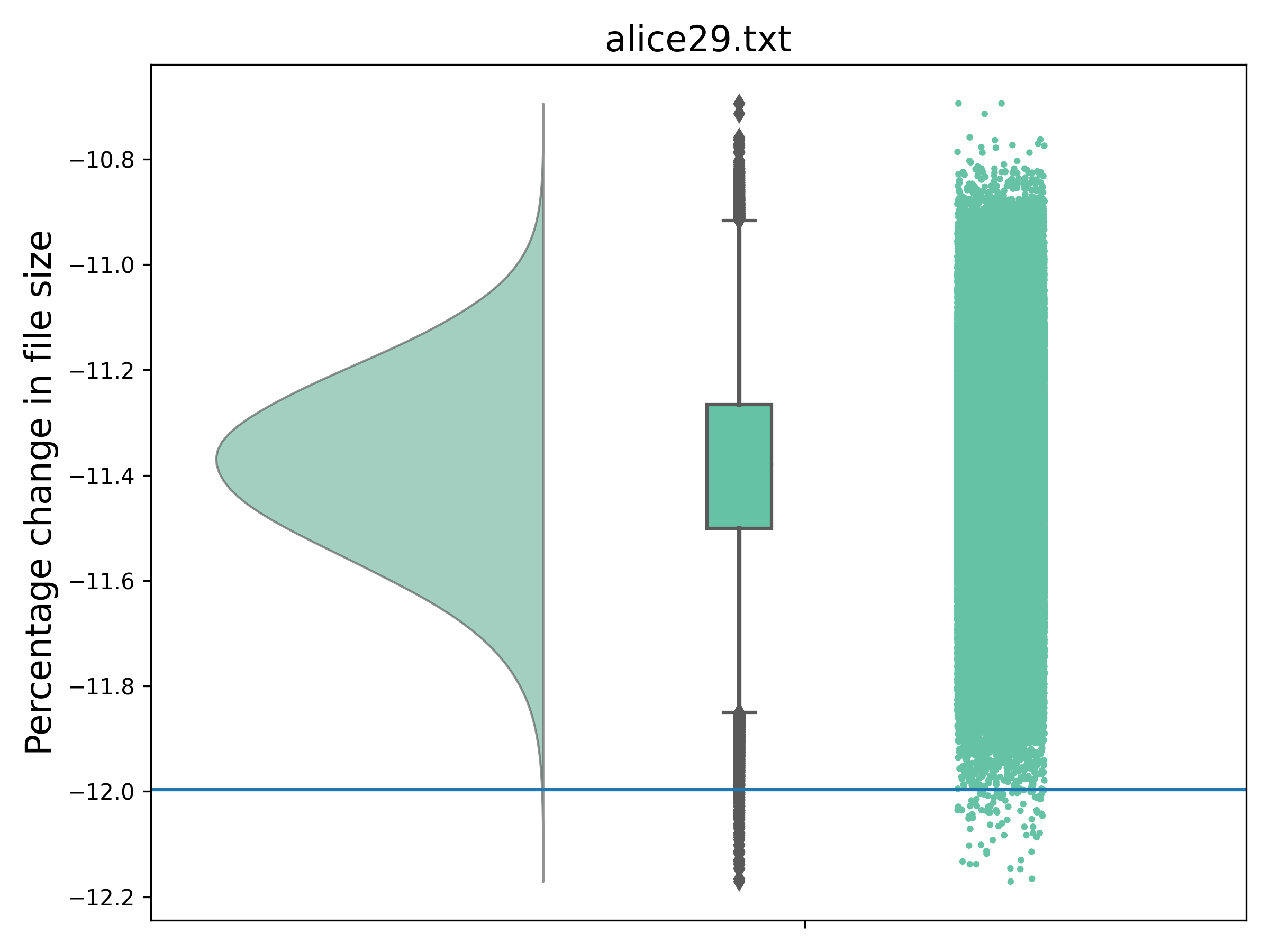}\\
        \includegraphics[width=0.6\textwidth]{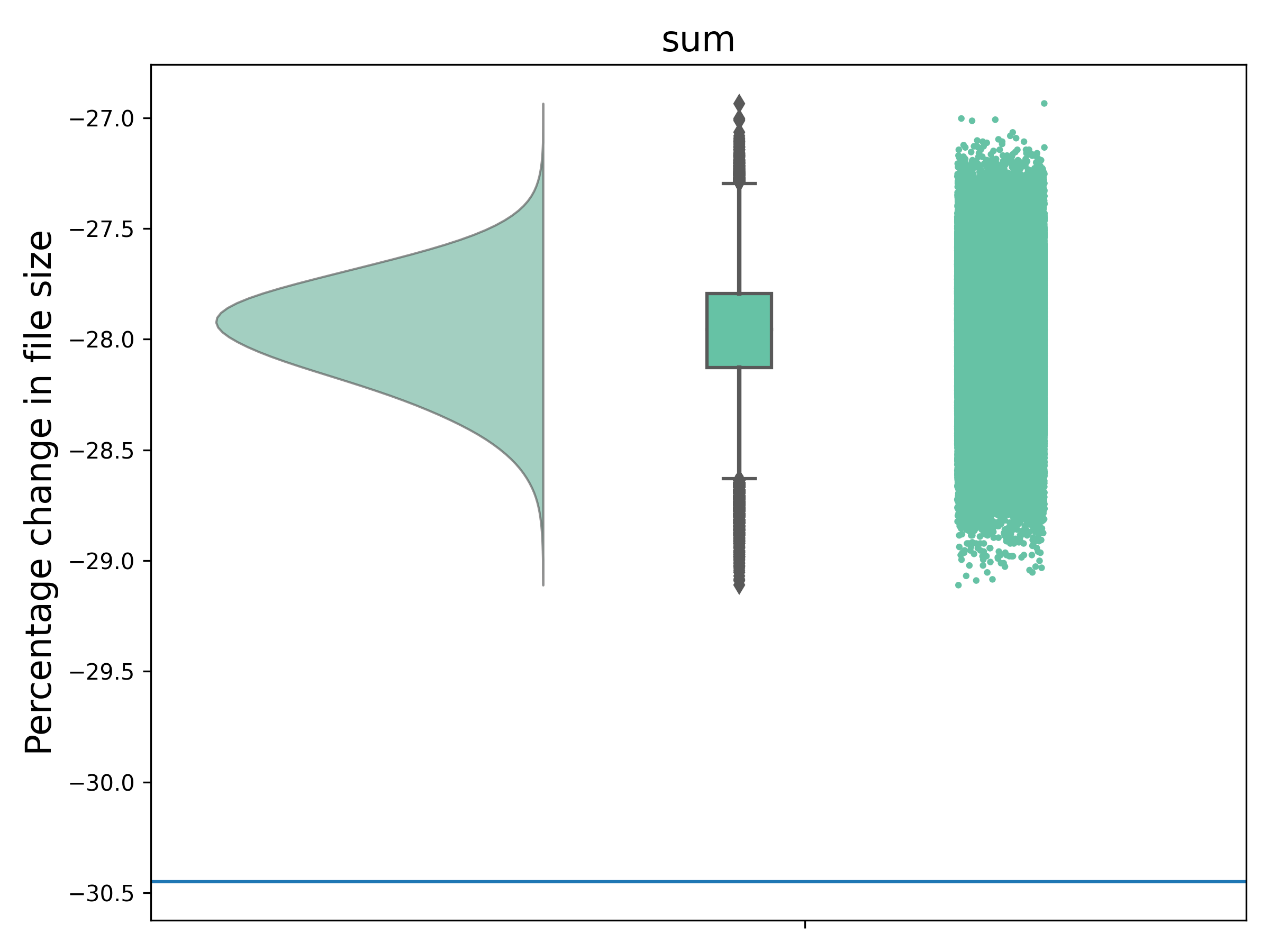}\\
        \includegraphics[width=0.6\textwidth]{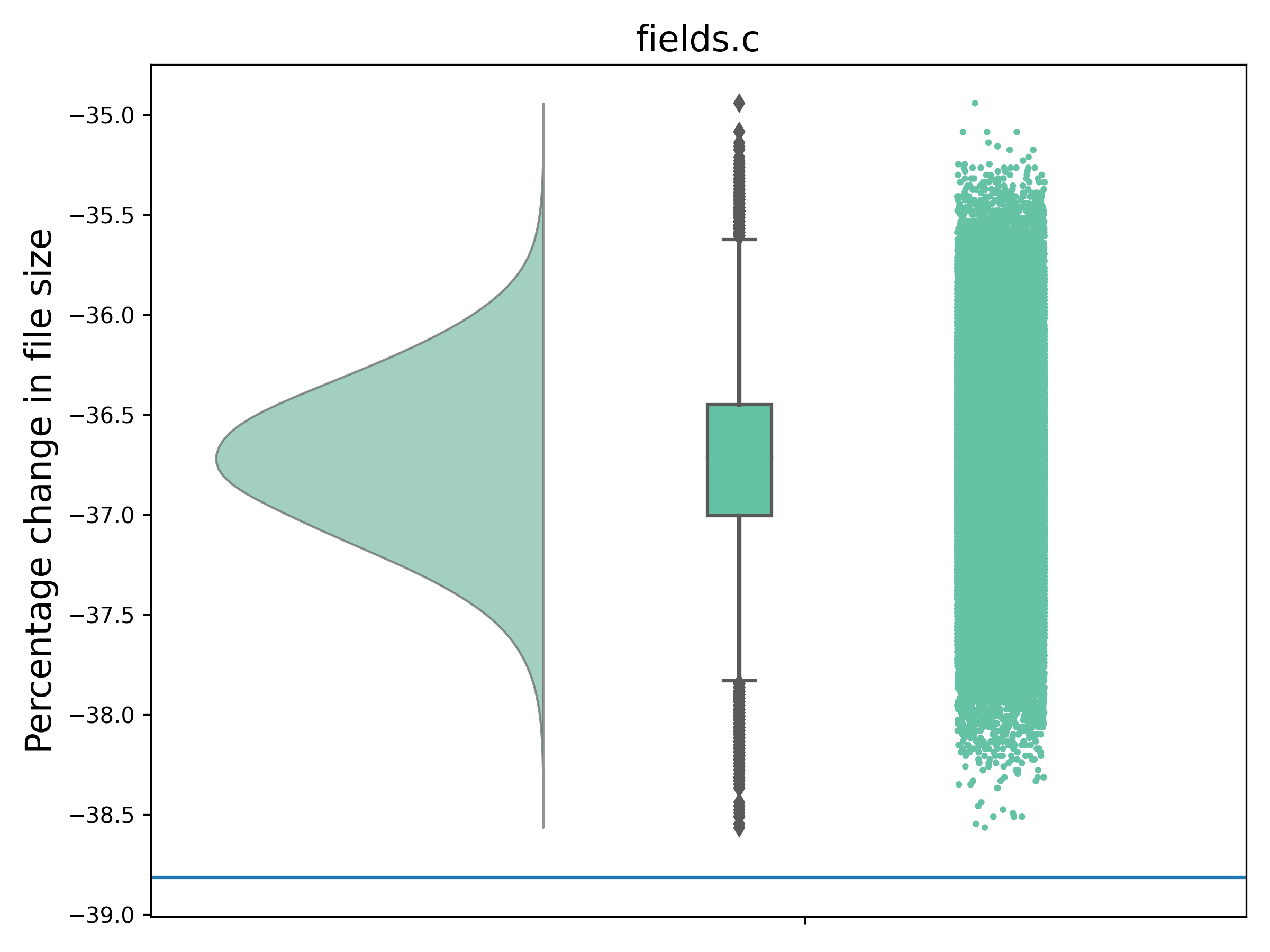}

    \caption{A raincloud plot of the percentage change in compression for 240,000 random samples of alphabet orderings used with the RLBWT for three of the corpus texts. Similar shaped distributions can be seen for the novel  \texttt{alice29.txt}, SPARC executable file \texttt{sum} and C source file \texttt{fields.c}. Horizontal blue line represents the ASCII ordering, which outperforms most of the randomly sampled orderings.}
    \label{fig:random-sample-cantrbry}
\end{figure}

However, most texts in the corpus can be compressed to be smaller than the original by using the RLBWT, despite the fact that most randomly chosen alphabets are poor choices (See Tab.~\ref{tab:random-sample-search-percent-change}). 
This is true even with the worst choice of alphabet order.
The best alphabet order found when randomly sampling for \texttt{ptt5}
reduced the file size by 74.207\%. 
The file \texttt{plrabn12.txt} does not compress well and the best alphabet order sampled for this file increased the size by 1.019\%, and ASCII performs worse than the best randomly sampled alphabet order.

While for most files, randomly sampled alphabet orderings may compress the size of the file somewhat, few sampled orderings improve on the ASCII alphabet ordering.
In fact, several of the initialization ordering methods (Sec \ref{sec:method-ls-init}) also already outperform even the best of the randomly sampled orderings for many of the files. 
The heatmap in Figure \ref{fig:heatmap_inits} shows the rankings of the initialization orderings for the different files. 
The benefits of ASCII and Chapin-Tate orderings can be seen clearly in this figure. 
It can also be seen from this figure that even after randomly sampling 240,000 orderings, the best of these is not good enough.
Random sampling is therefore a too costly search strategy and necessitates  another solution.

\begin{figure}
    \centering
    \includegraphics[width=0.9\textwidth]{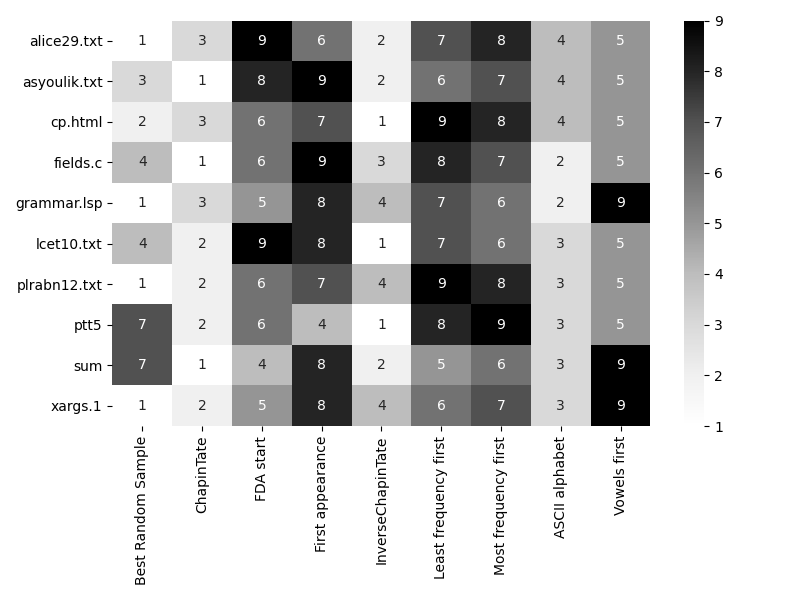}
    \label{fig:heatmap_inits}
    \caption{Ranking of each initialization method and the best randomly sampled ordering for each file before local search is applied.}
\end{figure}

\begin{table}
\centering
\rowcolors{1}{}{tablebackgroundcol}
\caption{Percentage difference compared to the original file size for each file when using the 240,000 orderings found using random sampling. The minimum, maximum, and mean percentage changes are shown for each text.}
\label{tab:random-sample-search-percent-change}
\tabularfix{
\begin{tabular}{|l|S[detect-weight, table-format=-2.3]|S[detect-weight, table-format=-2.3]|S[detect-weight, table-format=-2.3]|S[detect-weight, table-format=1.3]|}
\hline
File & {Min \% Change} & {Max \% Change} & {Mean} & {Std}\\
\hline
\texttt{alice29.txt}&-12.171&-10.694&-11.385&0.172\\
\texttt{asyoulik.txt}&-0.533&1.028&0.324&0.177\\
\texttt{cp.html}&-25.375&-23.156&-24.163&0.255\\
\texttt{fields.c}&-38.565&-34.942&-36.725&0.407\\
\texttt{grammar.lsp}&-28.353&-21.903&-24.955&0.711\\
\texttt{lcet10.txt}&-22.116&-21.14&-21.554&0.113\\
\texttt{plrabn12.txt}&1.019&2.073&1.633&0.109\\
\texttt{ptt5}&-74.207&-73.664&-73.914&0.080\\
\texttt{sum}&-29.111&-26.935&-27.967&0.251\\
\texttt{xargs.1}&-6.269&0.781&-2.484&0.735\\
\hline
\end{tabular}%
}
\end{table}

\subsection{Improved Fitness with Local Search}
\label{sec:results-local-search}

In contrast to the compression obtained via random sampling of alphabet orders (Sec.~\ref{sec:results-random-sampling}), even the most simple of our local search methods -- \textsc{Swap}, \textsc{Lex} -- can achieve a better fitness value for our tested data. 
Fig.~\ref{fig:random-sample-cantrbry-combined} shows the contrast between the best solutions found during a local search (using the \textsc{Swap} operator) and the results of the random sampling. This figure shows \texttt{alice29.txt}, \texttt{sum} and \texttt{fields.c}. 
Other corpus files have similar plots, which can be seen in our repository. The large gap between their distributions indicates that there are excellent alphabet orderings that have not previously been sampled, even when sampling 240,000 orderings.

\begin{figure}
    \centering
        \includegraphics[width=0.6\textwidth]{"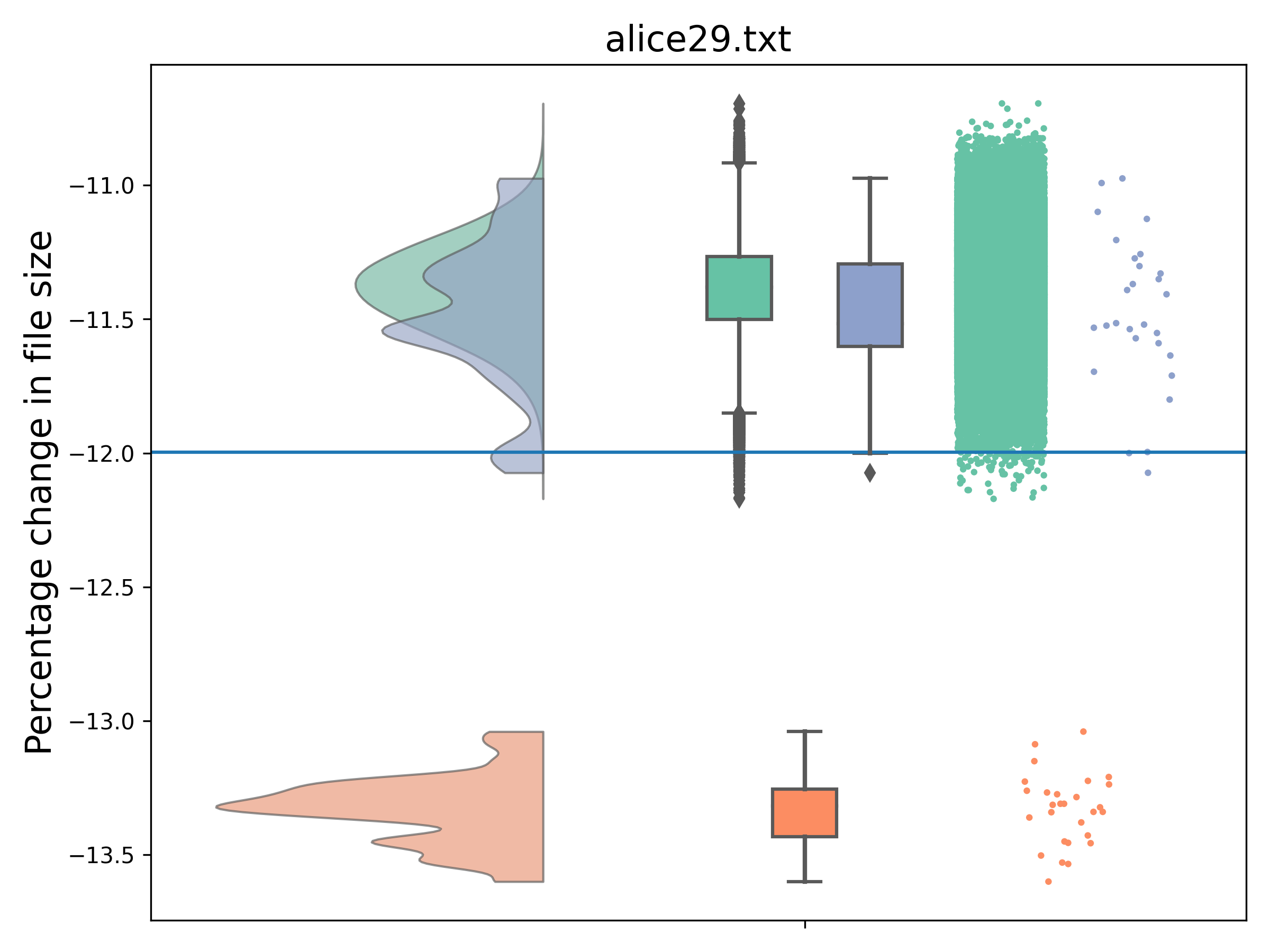"}
        \includegraphics[width=0.6\textwidth]{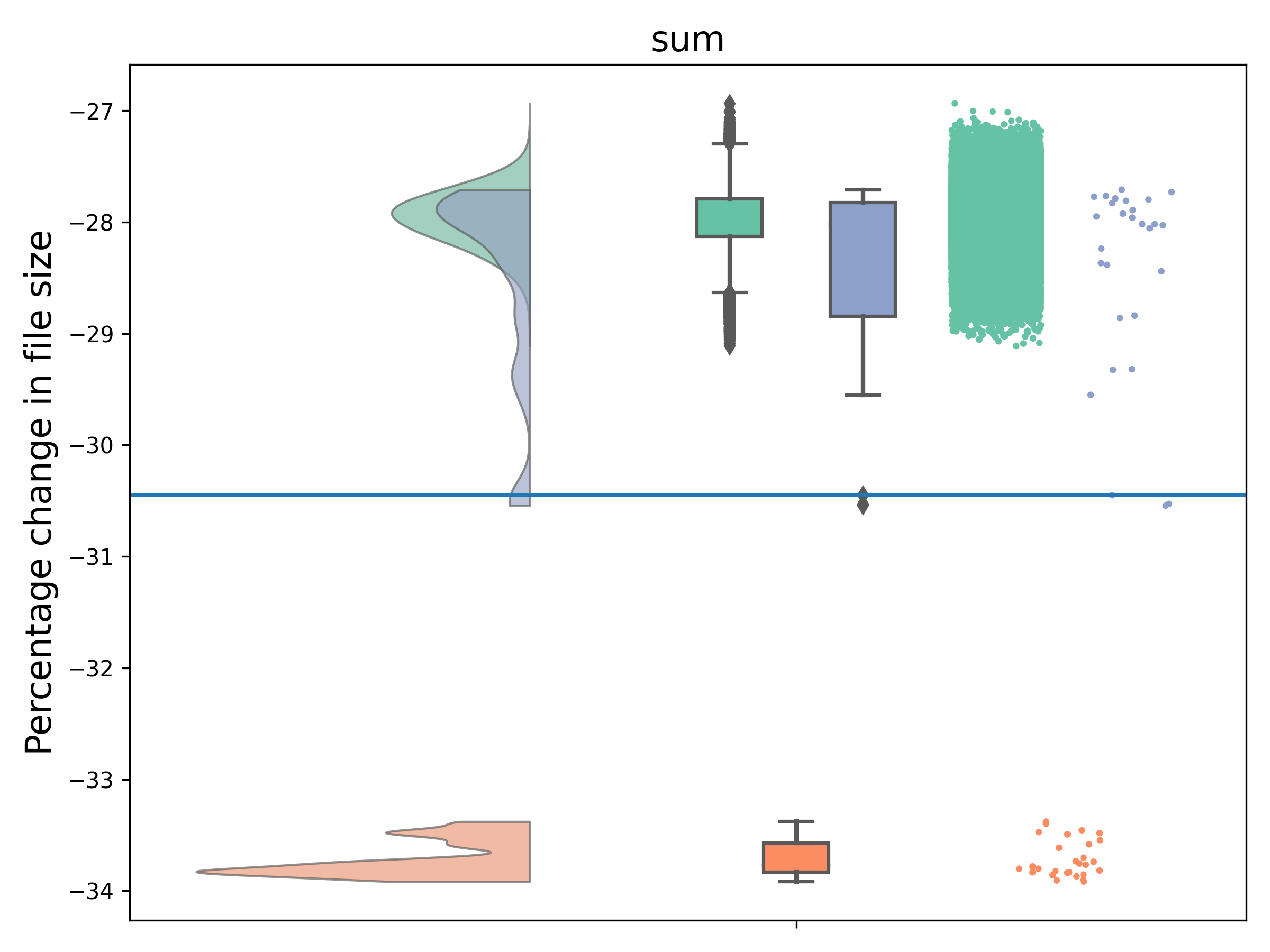}
        \includegraphics[width=0.6\textwidth]{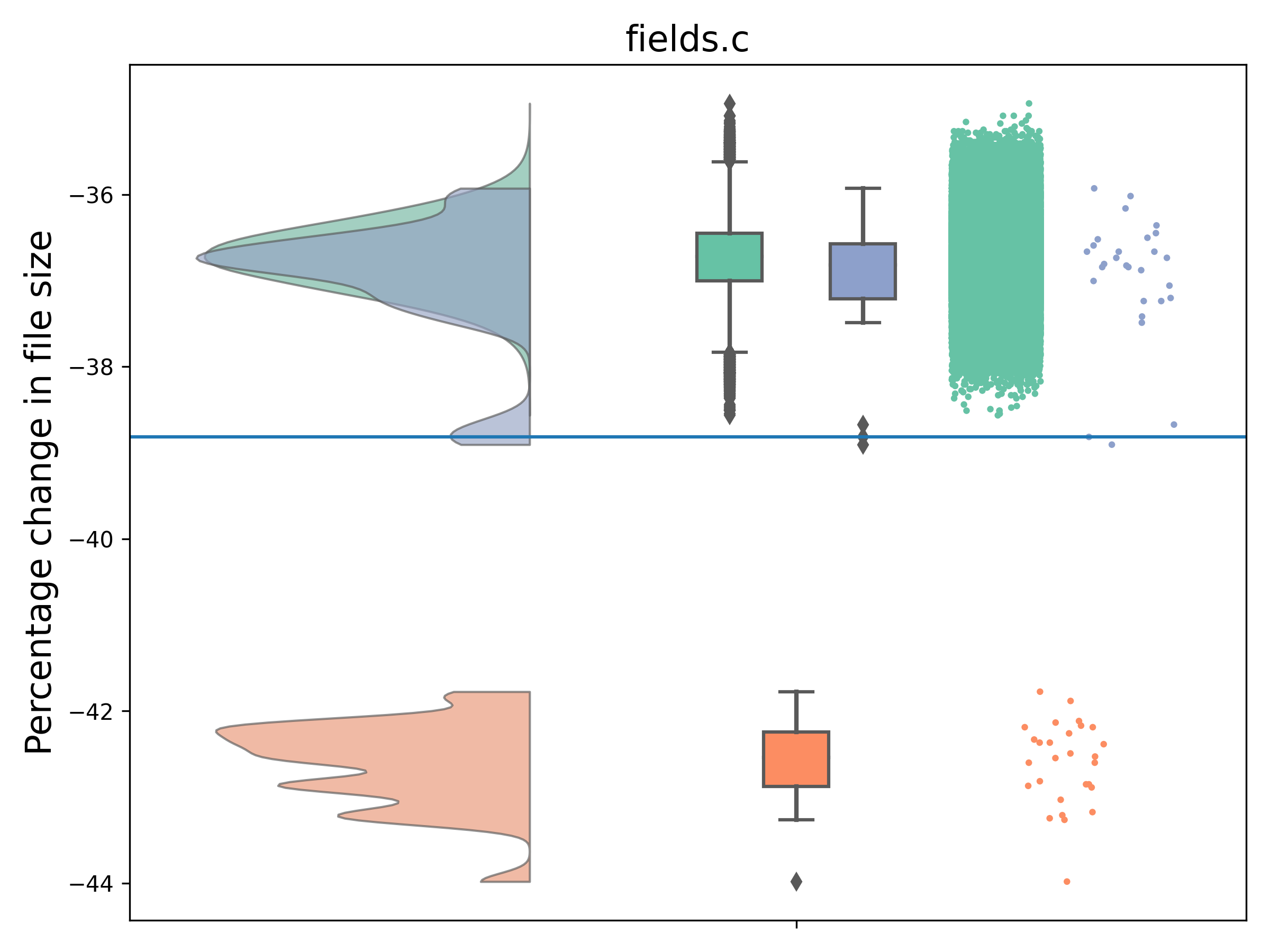}
    \caption{Raincloud plot showing that the best alphabet orders found at the conclusion of local search (orange) give noticeably better compression than that achieved using either randomly sampled alphabet orders (green) or the ASCII alphabet order (horizontal line). 
    The compression achieved by random samples of alphabet orderings are displayed in green. 
    The best achieved via local search with \textsc{Swap} only is shown in orange, (also see Section.~\ref{sec:results-local-search}). 
    The change in compression for the different local search initialization orderings at the start of the search are shown in blue, and these overlap with the random samples. 
    The compression when using the ASCII alphabet is plotted as a horizontal blue line and can be seen at the lower end of the random samples.}
    \label{fig:random-sample-cantrbry-combined}
\end{figure}

\subsection{The Impact of Initialization on Local Search}

\begin{table}
\centering
\rowcolors{1}{}{tablebackgroundcol}
\caption{Determining the best initialization orderings for an early-terminated search. Local search was performed with the \textsc{Swap, Lex} neighborhood for each corpus file, terminating after 1000 steps. IPCT=Inverse Permutation Chapin-Tate, CT=Chapin-Tate, Random=Random Initialization.}
\label{tab:table}
\label{tab:findingBestInits1k}
\tabularfix{
\begin{tabular}{|l|l|S[detect-weight, table-format=-1.3]|l|S[detect-weight, table-format=-1.3]|}
\hline
     & \multicolumn{2}{|c}{After 1000 steps} & \multicolumn{2}{|c|}{At local minimum} \\
File & {Best initialization (1000)} & {C (1000)} & {Best initialization (all)} & {C (all)}\\
\hline
\texttt{alice29.txt}&IPCT&-12.368&Random&-13.601\\
\texttt{asyoulik.txt}&CT&-1.108&CT&-2.07\\
\texttt{cp.html}&IPCT&-25.92&Random&-27.993\\
\texttt{fields.c}&CT&-40.359&ASCII&-43.982\\
\texttt{grammar.lsp}&Random&-29.589&Random&-33.996\\
\texttt{lcet10.txt}&IPCT&-22.503&Random&-23.04\\
\texttt{plrabn12.txt}&ASCII&0.948&Random&0.228\\
\texttt{ptt5}&IPCT&-74.472&IPCT&-74.748\\
\texttt{sum}&ASCII&-30.737&FDA&-33.917\\
\texttt{xargs.1}&CT&-7.783&CT&-12.042\\
\hline
\end{tabular}%
}
\end{table}

When performing a local search, the initial alphabet order makes a difference to the best solution that can be found by the local search using the \textsc{Swap} operator.

Tab.~\ref{tab:findingBestInits1k} shows which initialization ordering achieves the best compression when the search terminates at the local minimum or after 1000 steps.
While randomly chosen orders are competitive if the search is terminated early after 1000 steps, ASCII, Chapin-Tate, and Inverse Permutation Chapin-Tate perform best if the search is allowed to complete.

The plots in Fig.~\ref{fig:over-time-local-search-starts} exemplify how the results improve over the time taken by the search. 
A steep drop in the size of the file is observed followed by a large number of steps until the local optima is reached. 
The initialization order that leads to the best local minimum at the end of the search is not obvious at the start of the search, nor is there consistently an initialization order that would produce a good local minimum across all files.
The range of random sample fitness may completely cover the fixed start positions as in \texttt{alice29.txt} or lay above many of the best fitness starts as in \texttt{fields.c}. Overall the range of random sample starts covers a large amount of solutions that are found using the fixed start positions in our tested files.

The steep improvement early in the search suggests that a limited search may still be beneficial. 
If the search is terminated early, ASCII and Chapin-Tate are the best orderings from which to start the search (Fig.~\ref{fig:over-time-local-search-starts-limited}). 

While using a random initialization still yields an improvement in fitness over time, the overall change in fitness is not as good as choosing a fixed initialization for the explored texts.
It may be possible for multiple orderings to achieve around the same fitness for 1000 neighbor evaluations.

\begin{figure}
    \centering
    \includegraphics[width=0.6\textwidth]{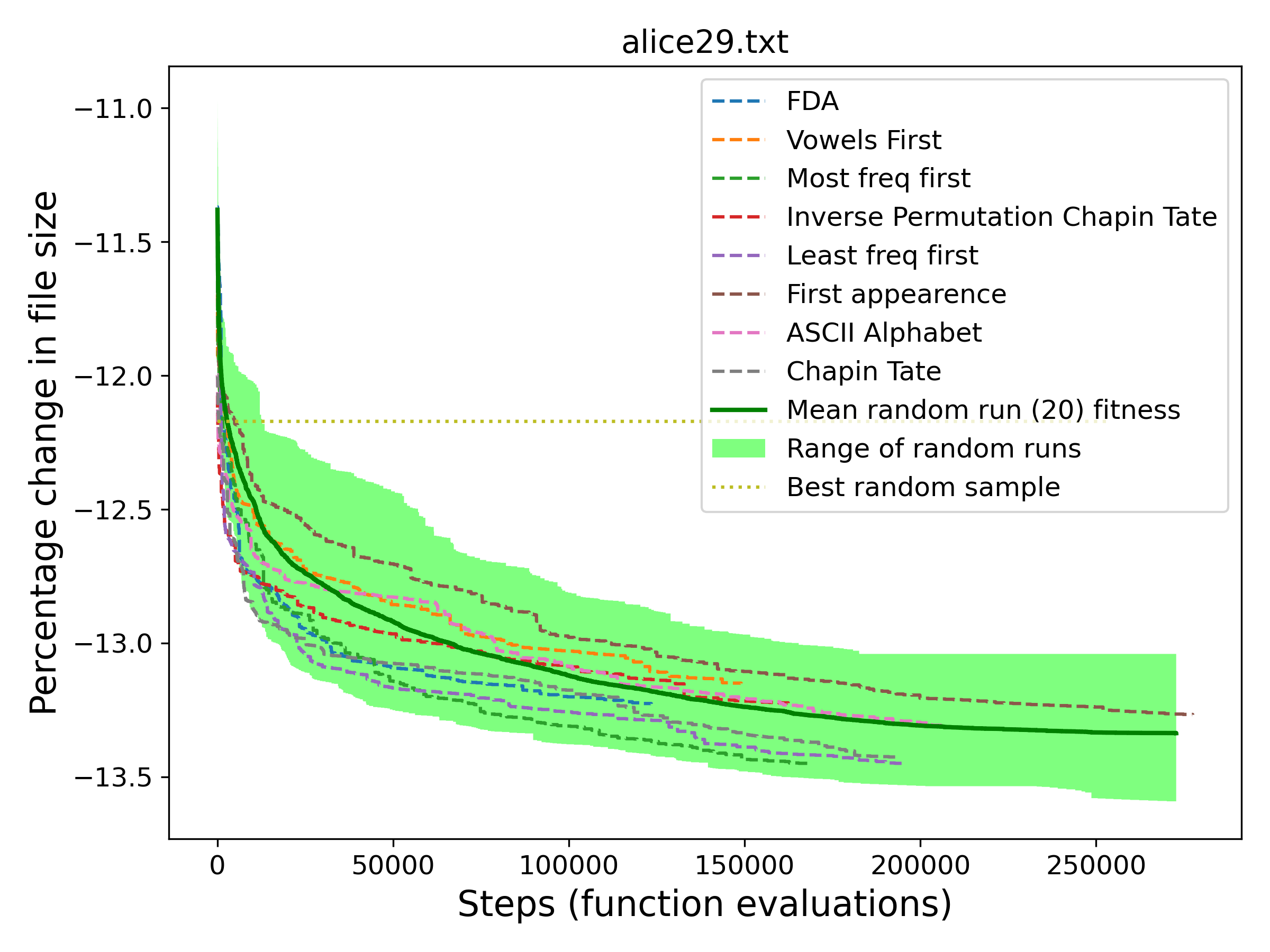}
    \includegraphics[width=0.6\textwidth]{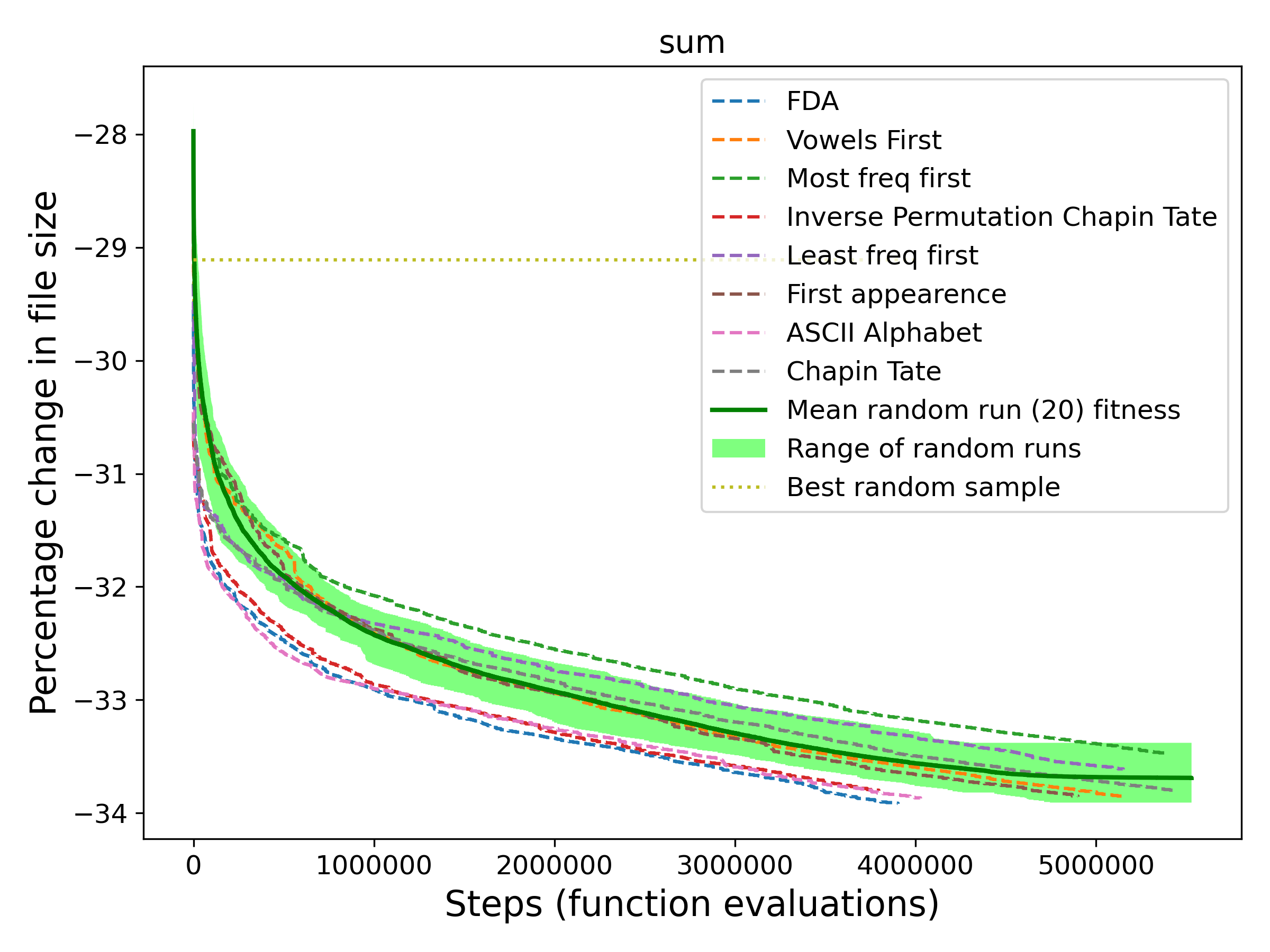}
    \includegraphics[width=0.6\textwidth]{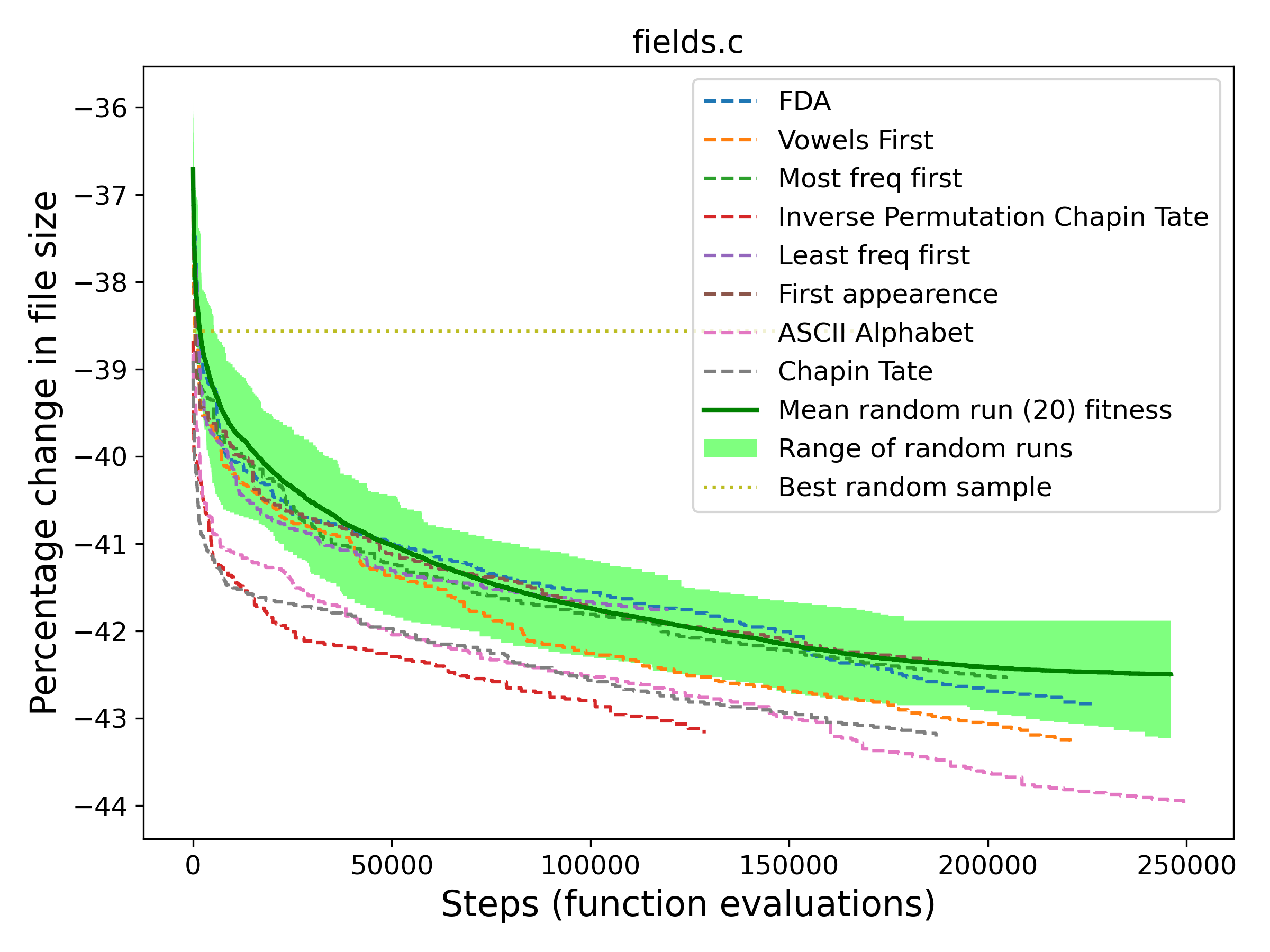}
    \caption{Local search using \textsc{Swap, Lex} from different alphabet order initialization methods over time until a local minimum is reached. It can be observed that there is no best initial order that consistently results in the best local minima for all texts.}
    \label{fig:over-time-local-search-starts}
\end{figure}

From this we conclude that the landscape has lots of local optima and that the path through the landscape is therefore important. We examine other methods such as reverse and random neighbor orderings, and using \textsc{Insert} to make further jumps in the landscape.

\begin{figure}
    \centering
    \includegraphics[width=0.6\textwidth]{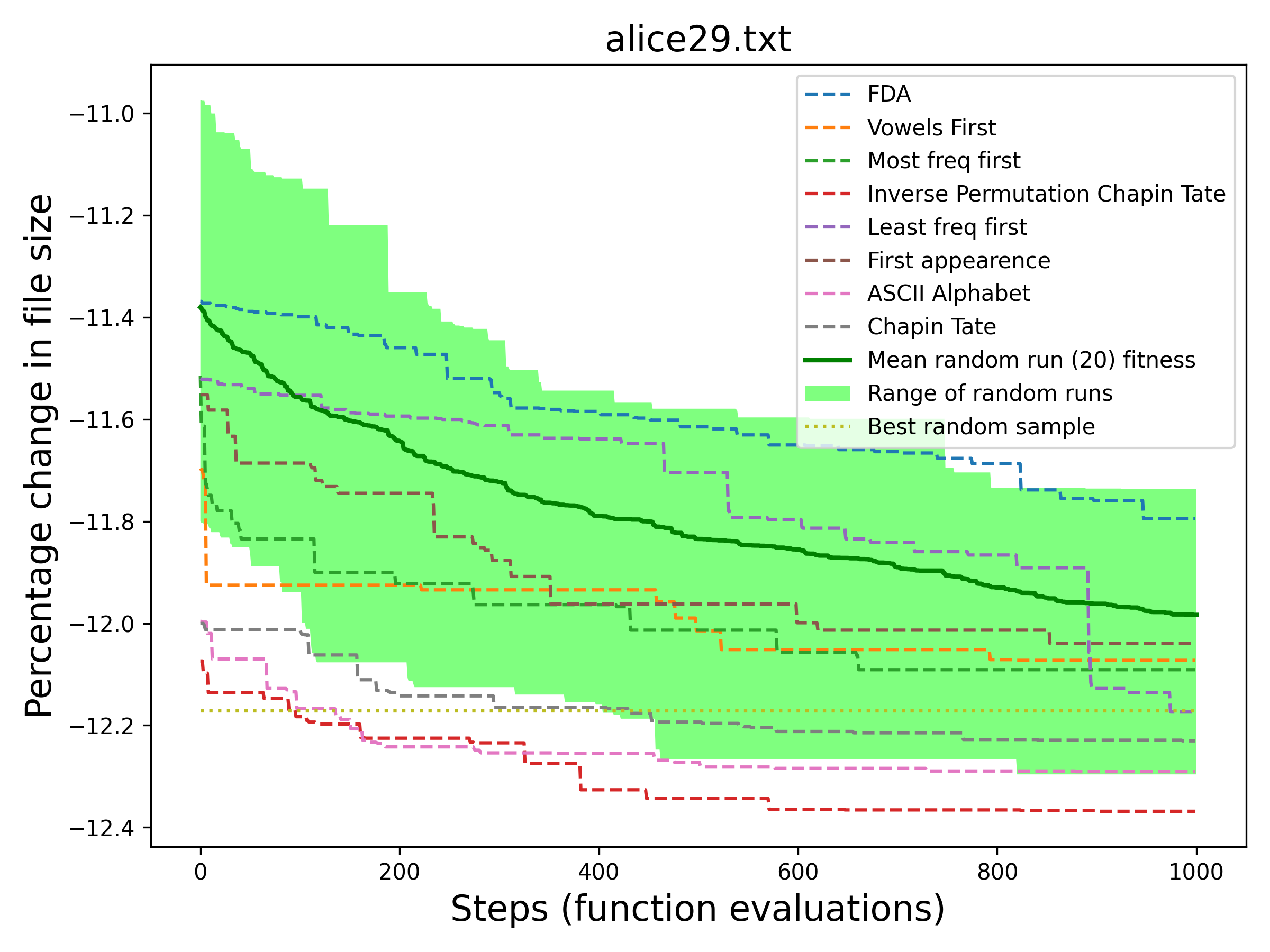}\\
    \includegraphics[width=0.6\textwidth]{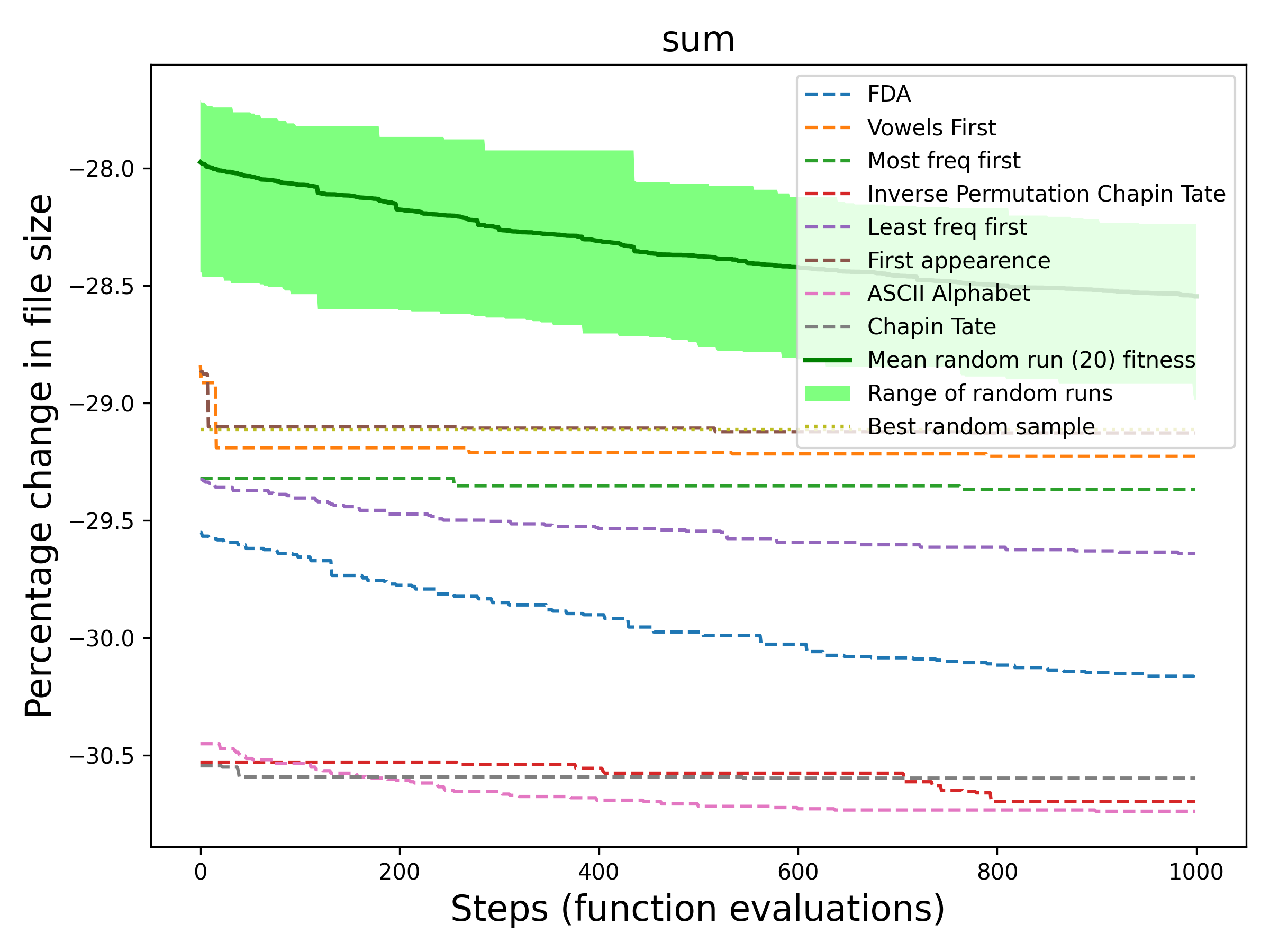}\\
    \includegraphics[width=0.6\textwidth]{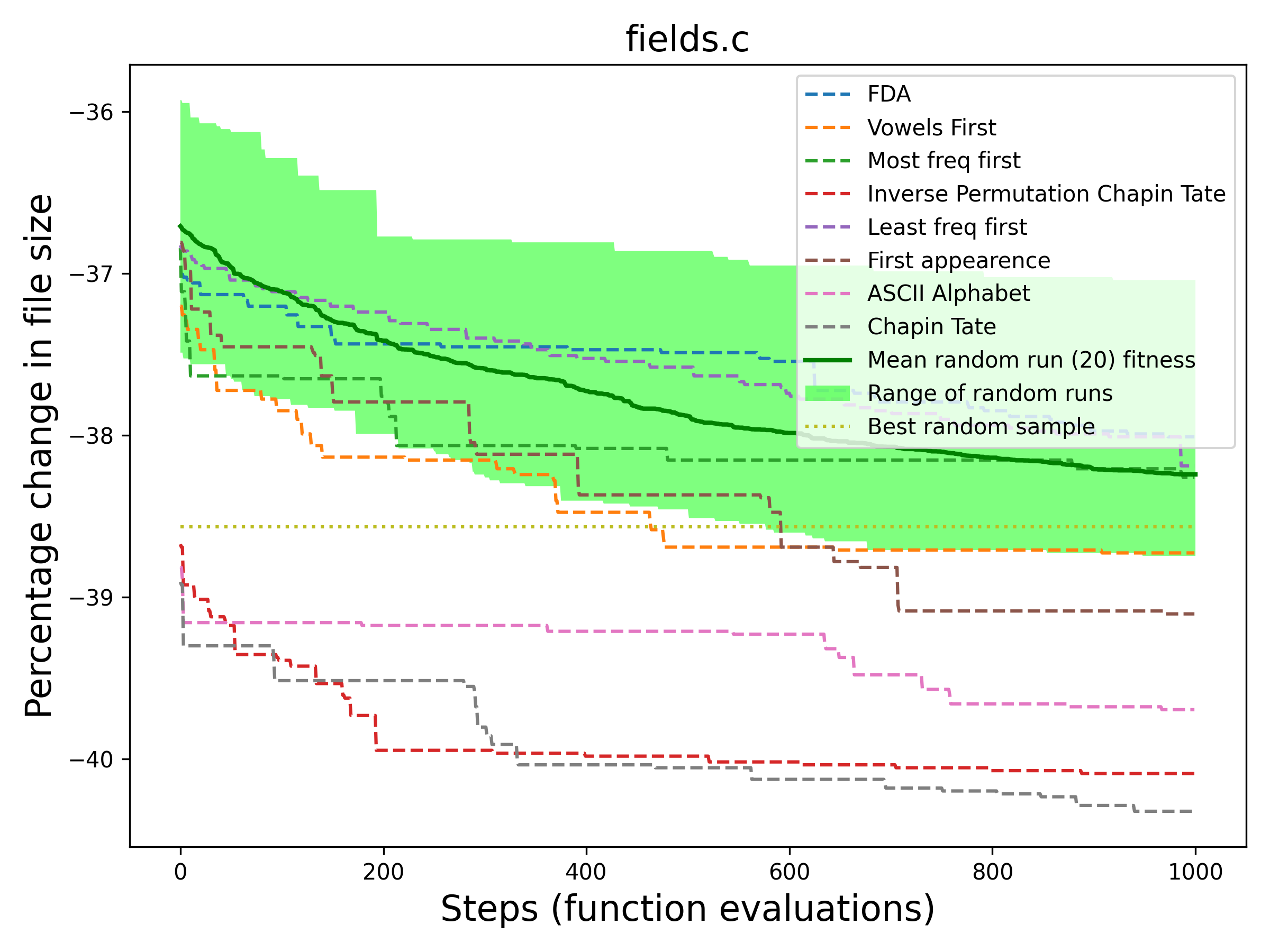}\\
    \caption{Local search for different alphabet order initialization methods over time with \textsc{Swap, Lex}, limited to 1000 neighbor evaluations.
    It can be observed that there is no best initial order that consistently results in the best local minima for all texts, and that the best initialization order is dependent on the number of steps performed.}
    \label{fig:over-time-local-search-starts-limited}
\end{figure}

\subsection{Local Search Operators: \textsc{Swap} and \textsc{Insert}}
\label{sec:whichOperatorToChoose}

The path that the local search explores through the search space is important in determining where in the space the local optima is found. Evaluating more local search neighbors may therefore lead to finding better neighbors and a better overall optima. Since \textsc{Swap} moves two elements into new locations at once.
We therefore consider the \textsc{Insert} operator, either in combination with \textsc{Swap} (as first or second operator) or alone.

Tab.~\ref{tab:local-search-steps-to-beat-random} shows the number of steps taken using each of the different search operators to find an ordering that performs better than the best of the 240,000 randomly sampled orderings. 
We note that local search only requires a few steps before outperforming random sampling, independently of any search operators, and typically takes only a few seconds to compute.
Indeed, table \ref{tab:random-sampling-updates} shows that methods based on local search create a large number of updates when compared to a naive random sampling of the search space. While the average number of updates obtained with the latter is close to the theoretical one (12.4 vs. 12.38), we see that local search algorithms provide substantially more updates. For instance, the \textsc{Swap} method using the lexicographic order generates on average 236 updates across 10 files. Such a number of updates with random sampling would on expectation require more than $10^{100}$ samples  ($\approx e^{236.5}$), which is today not computable.

Methods involving randomized neighbor orderings perform well in the time taken to beat the best of the 240,000 random samples compared to all other methods.
Of these, the \textsc{Swap then Insert} operator performs best.

\begin{table}
\centering
\rowcolors{1}{}{tablebackgroundcol}
\caption{The minimum number of local search steps for the best initialization to find an ordering performing better than the best ordering out of the 240,000 randomly sampled orderings. Files \texttt{asyoulik.txt}, \texttt{cp.html}, \texttt{fields.c}, \texttt{ptt5}, \texttt{lcet10.txt}, and \texttt{sum} are not shown because no local search steps were needed, as one of the initial orderings was already better than the randomly sampled orderings without the need for search.
I=\textsc{Insert}, S=\textsc{Swap}, ItS=\textsc{Insert} then \textsc{Swap}, StI=\textsc{Swap} then \textsc{Insert}.
}
\label{tab:local-search-steps-to-beat-random}
\tabularfix{
\begin{tabular}{|l|S[detect-weight, table-format=3]|S[detect-weight, table-format=3]|S[detect-weight, table-format=3]|S[detect-weight, table-format=3]|S[detect-weight, table-format=3]|S[detect-weight, table-format=3]|}
\hline
Method & {\texttt{alice29.txt}} & {\texttt{grammar.lsp}} & {\texttt{plrabn12.txt}} & {\texttt{xargs.1}} & {All Others} \\
\hline
\textsc{I, Lex}&1437&107&9254&553&0\\
\textsc{I, Random}&10&17&40&22&0\\
\textsc{I, RevLex}&57&1475&1009&1249&0\\
\textsc{ItS, Lex}&1437&107&9254&553&0\\
\textsc{ItS, Random}&16&15&33&54&0\\
\textsc{ItS, RevLex}&57&1475&1009&1249&0\\
\textsc{S, Lex}&96&87&302&178&0\\
\textsc{S, Random}&4&34&16&41&0\\
\textsc{S, RevLex}&105&421&235&428&0\\
\textsc{StI, Lex}&96&87&302&178&0\\
\textsc{StI, Random}&16&20&37&41&0\\
\textsc{StI, RevLex}&105&421&235&428&0\\
\hline
\end{tabular}%
}
\end{table}

\begin{table}
\centering
\rowcolors{1}{}{tablebackgroundcol}
\caption{Number of successive improvements obtained with two local search algorithms compared to random sampling for ASCII ordering. Total number of steps are within brackets. Random sampling has created on average 12.4 updates across the 10 files while the theoretical mean number of updates for random sampling is approximately 12.38 for 240K samples (See section \ref{sec:method-sampling}). In contrast, local search algorithms deliver significantly more updates to the compression when compared to random sampling.}
\tabularfix{
\begin{tabular}{|l|r|r|r|}
\hline
File & {Random Sampling} & \texttt{Swap, Lex} & \texttt{Insert, Lex}\\
\hline
\texttt{alice29.txt}&14 (240K)&231 (205.84K)&319 (519.2K)\\
\texttt{asyoulik.txt}&11 (240K)&218 (125.32K)&310 (399.06K)\\
\texttt{cp.html}&17 (240K)&176 (242.68K)&241 (491.92K)\\
\texttt{fields.c}&10 (240K)&157 (249.31K)&169 (496.97K)\\
\texttt{grammar.lsp}&11 (240K)&67 (65.54K)&86 (161.88K)\\
\texttt{lcet10.txt}&11 (240K)&259 (340.97K)&380 (934.5K)\\
\texttt{plrabn12.txt}&15 (240K)&383 (435.92K)&457 (929.79K)\\
\texttt{ptt5}&12 (240K)&352 (1.25M)&549 (3.59M)\\
\texttt{sum}&10 (240K)&452 (4.03M)&561 (10M)\\
\texttt{xargs.1}&13 (240K)&70 (63.98K)&102 (192.87K)\\
\hline
average number of updates&12.4&236.5&317.4\\
\hline
\end{tabular}%
}
\label{tab:random-sampling-updates}
\end{table}

However, to fully locate any local minimum may in some cases take a very long time (Fig.~\ref{fig:over-time-local-search-compare-methods}). 
Our experiments are limited to 10 million steps for any single run, which is reached for some configurations.
When considering the initialization that locates the best fitness for any method, we find an initial steep drop in fitness for all methods, and that there are groupings of methods that perform similarly.
Generally the methods involving randomized neighbor orderings perform well in few steps and have a fitness which remains competitive with the \textsc{Lex} and \textsc{RevLex} methods.
We observe that even when the number of steps is limited to relatively few in comparison to the number needed to reach a local minimum, the percentage change in file size reached may still be good.

\begin{figure}
    \centering
    \includegraphics[width=0.6\textwidth]{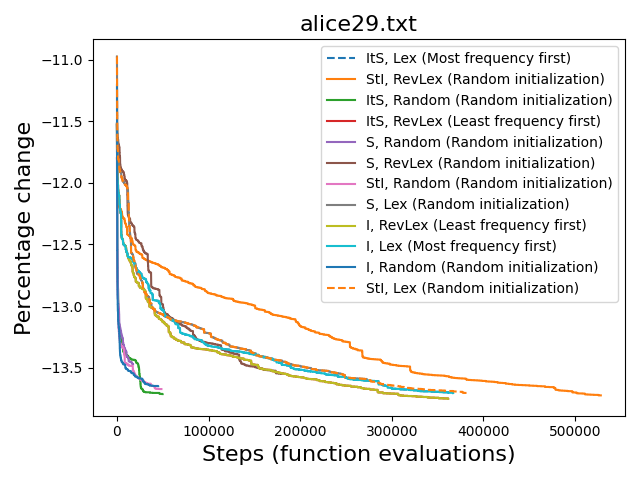}\\
    \includegraphics[width=0.6\textwidth]{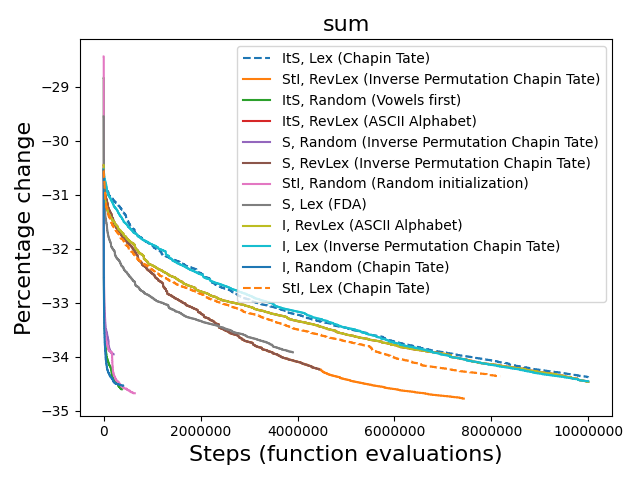}\\
    \includegraphics[width=0.6\textwidth]{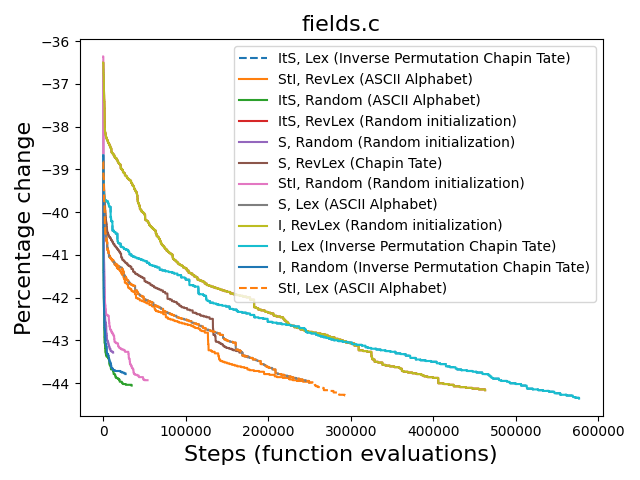}\\
    \caption{The best initialization at the minimum for corpus texts \texttt{alice29.txt}, \texttt{sum}, and \texttt{fields.c} over time for different neighborhood search methods.}
    \label{fig:over-time-local-search-compare-methods}
\end{figure}

This is further demonstrated when considering the final fitness achieved in all runs across all methods. Across our tested files, the number of required steps to locate a minimum splits the methods into three groups (in increasing order of steps): random neighbor ordering methods, \textsc{Swap} and \textsc{Swap}-first methods, and \textsc{Insert} and \textsc{Insert}-first methods (Fig.~\ref{fig:combined-steps-plots-local-search}). Even when considering the final fitness achieved, a randomized neighbor ordering remains competitive (Fig.~\ref{fig:combined-fitness-plots-local-search}).
The file \texttt{sum} is not shown as it has data which was not run to completion due to prohibitive runtimes (Section.~\ref{sec:methods-setup}). However, the trend described also holds with the completed data for \texttt{sum}.

When completing the search to a local minimum, \textsc{Insert} and \textsc{Insert}-first methods perform slightly better than methods involving \textsc{Swap}, however the \textsc{Insert} and \textsc{Insert}-first methods may take prohibitively long compared to \textsc{Swap}, as they search a wider neighborhood, for very little gain.

\begin{figure}
    \centering
    \includegraphics[width=0.6\textwidth]{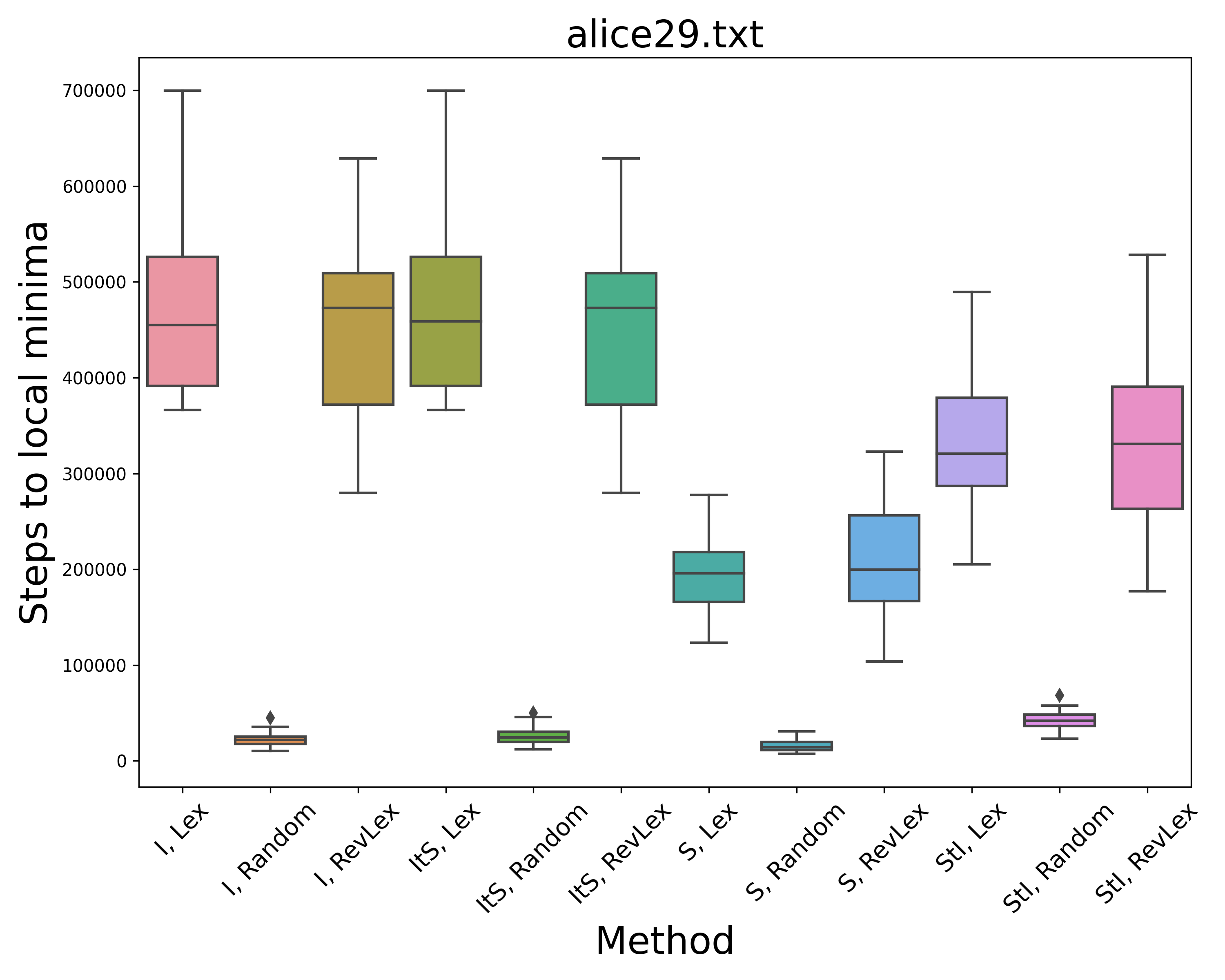}
    \includegraphics[width=0.6\textwidth]{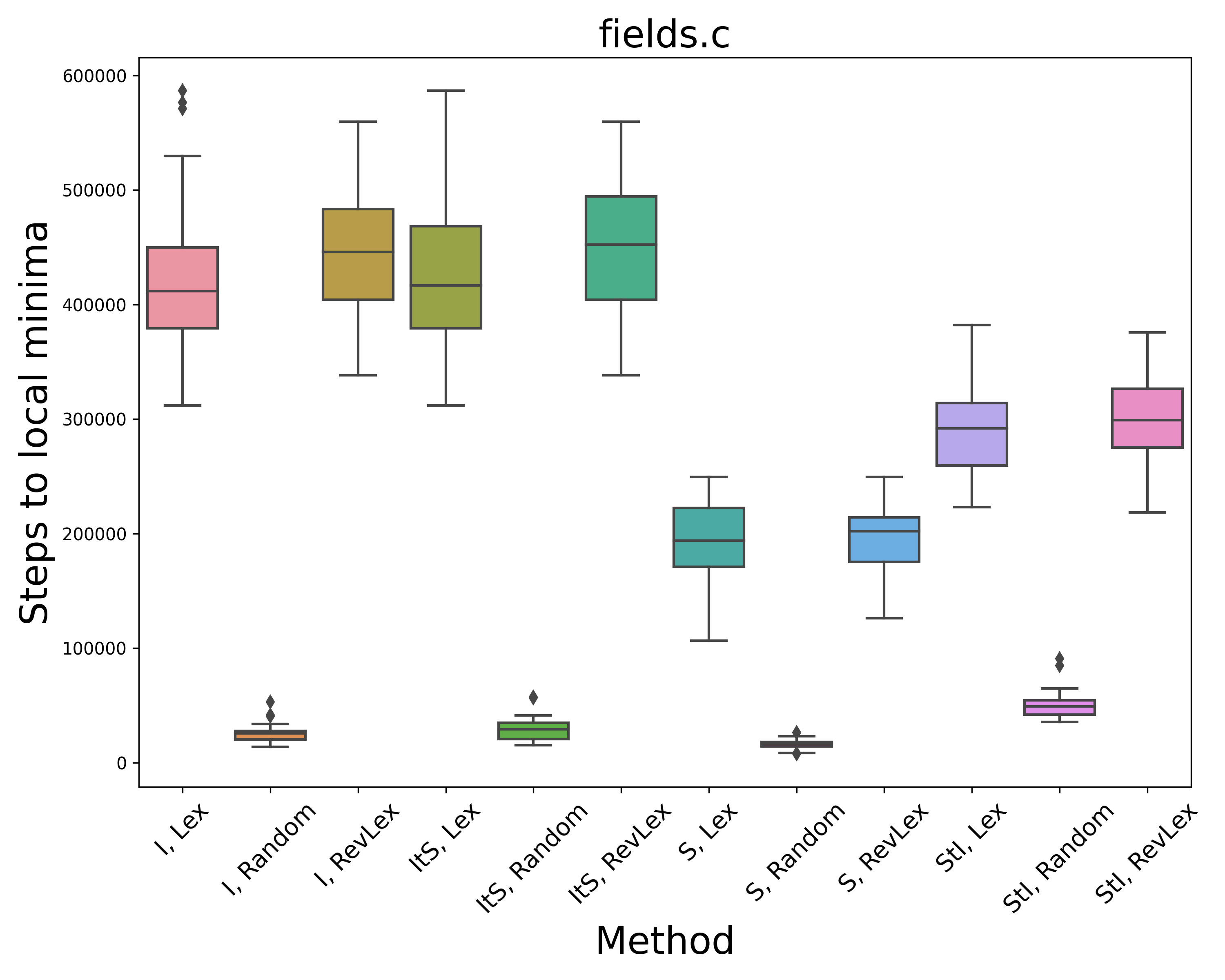}
    \includegraphics[width=0.6\textwidth]{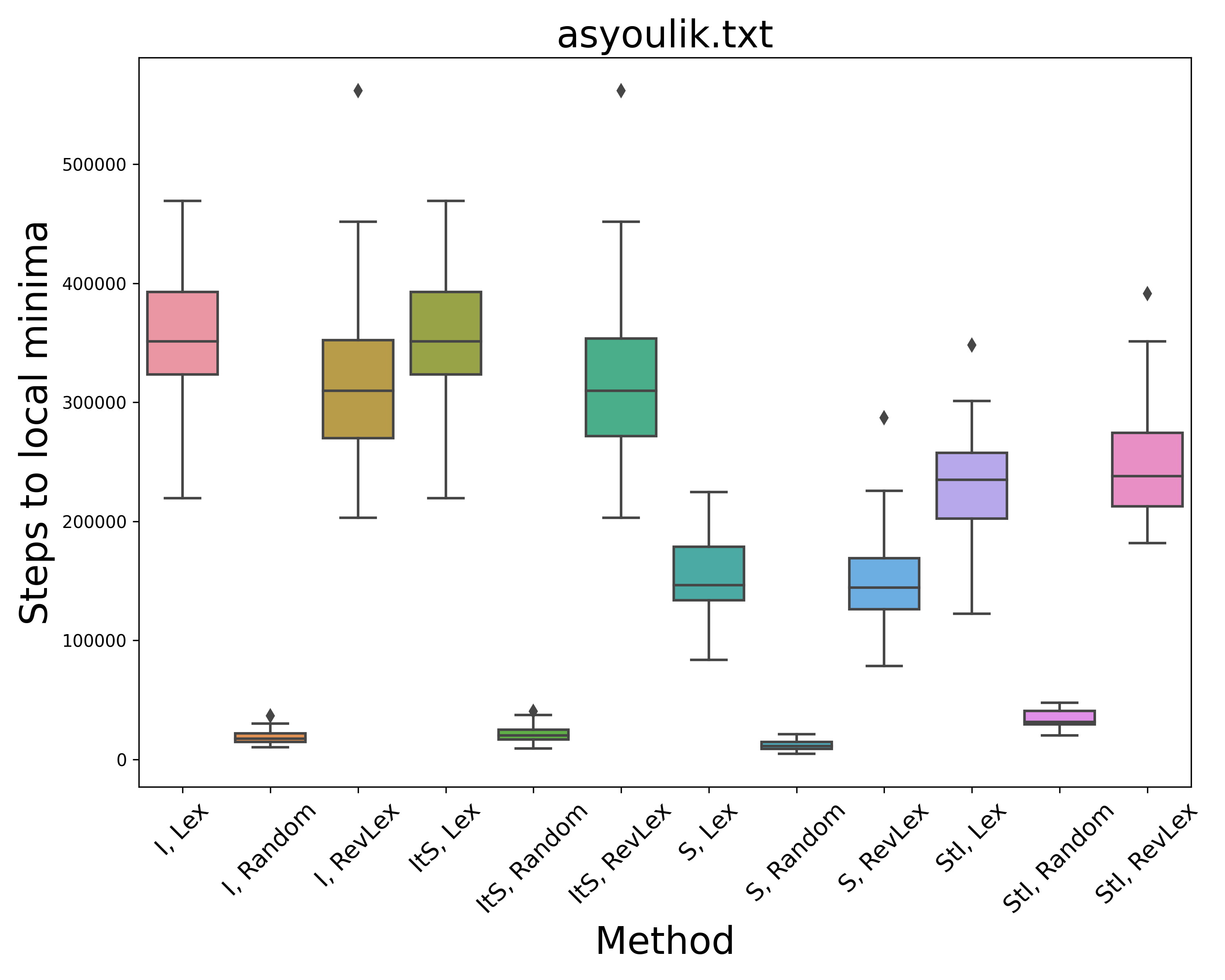}
    \caption{Number of steps taken to find a local minimum for \texttt{alice29.txt}, \texttt{fields.c}, and \texttt{asyoulik.txt} for each neighborhood method in local search. The variation in each box plot shows the difference made by distinct initializations.}
    \label{fig:combined-steps-plots-local-search}
\end{figure}

\begin{figure}
    \centering
    \includegraphics[width=0.6\textwidth]{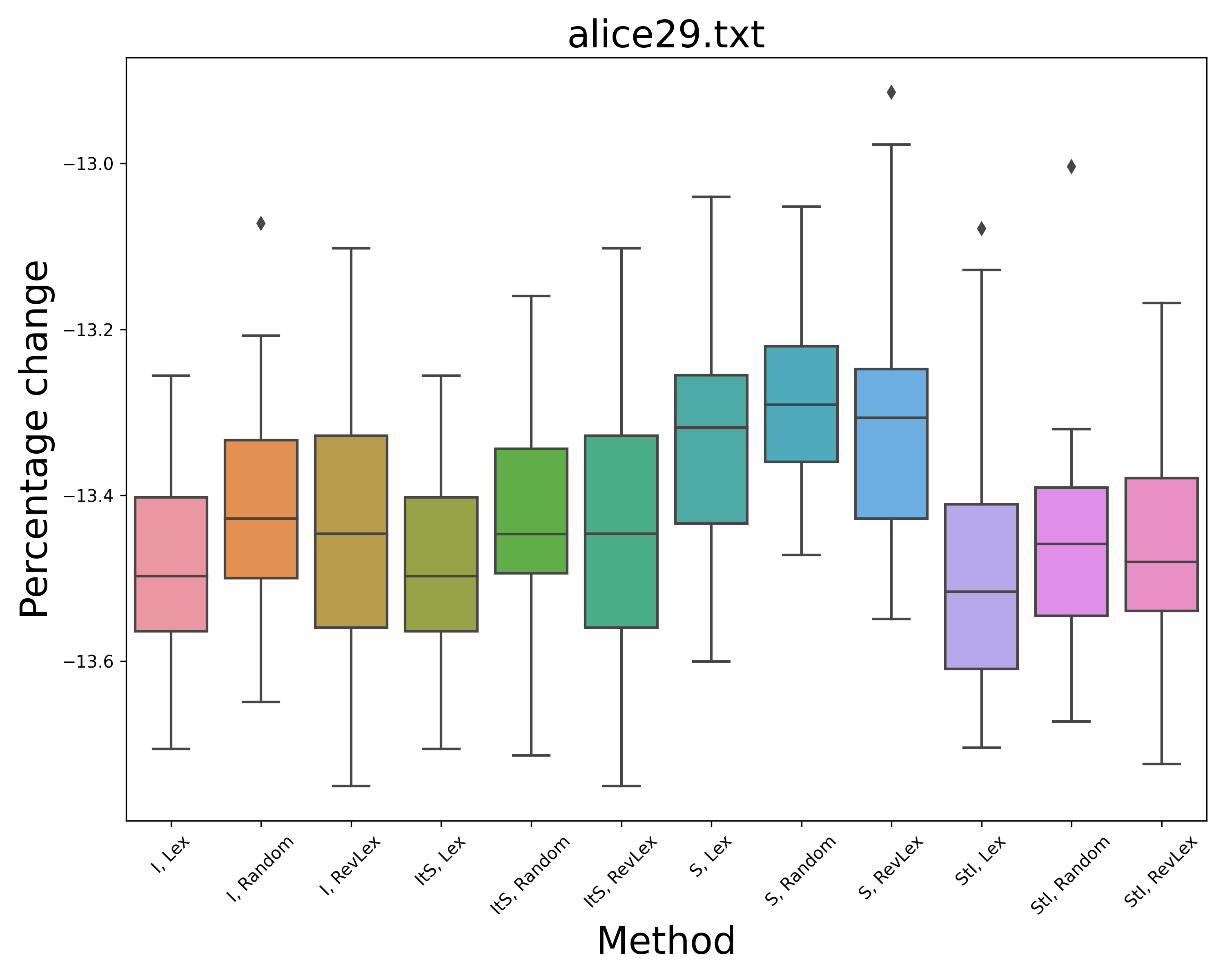}
    \includegraphics[width=0.6\textwidth]{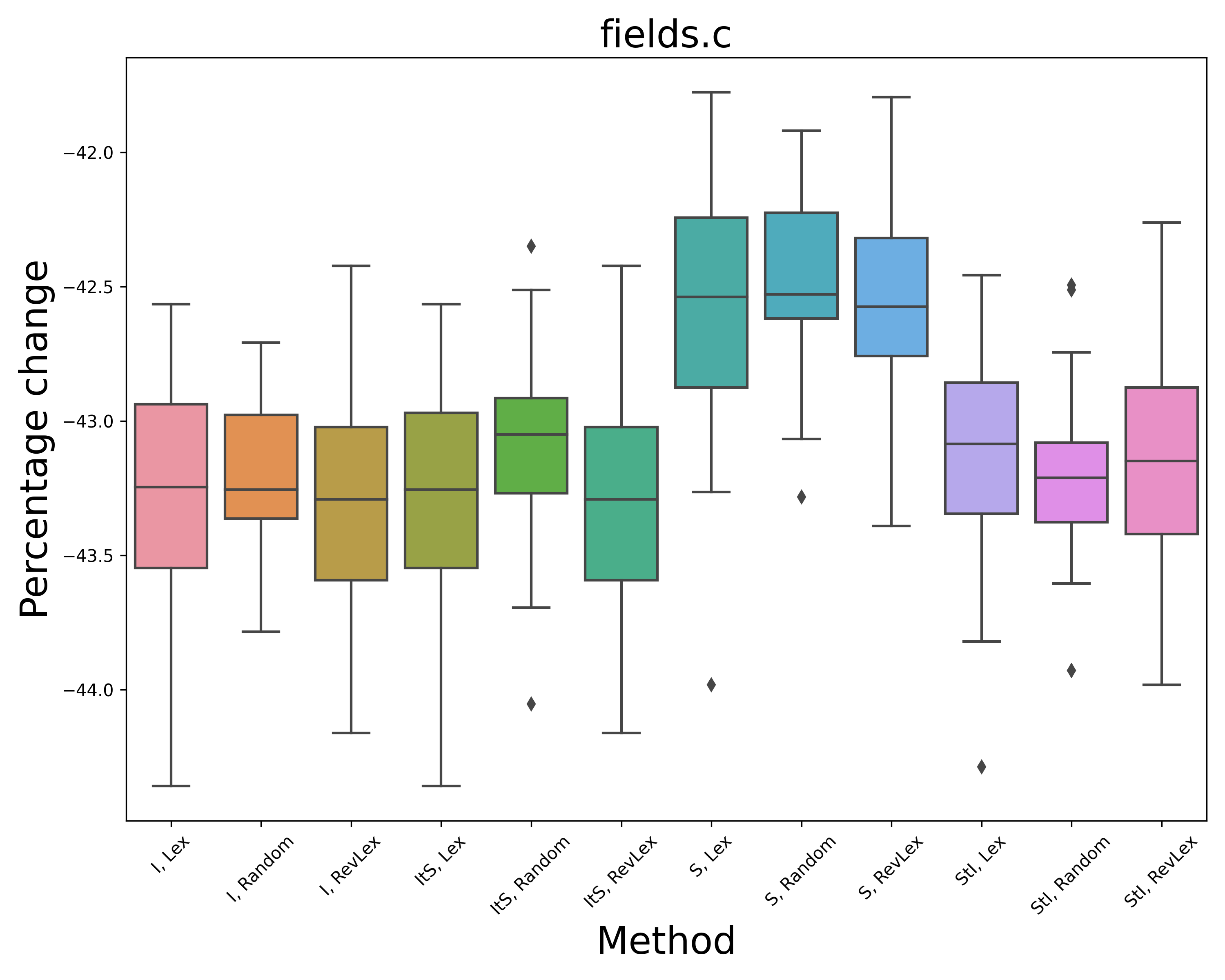}
    \includegraphics[width=0.6\textwidth]{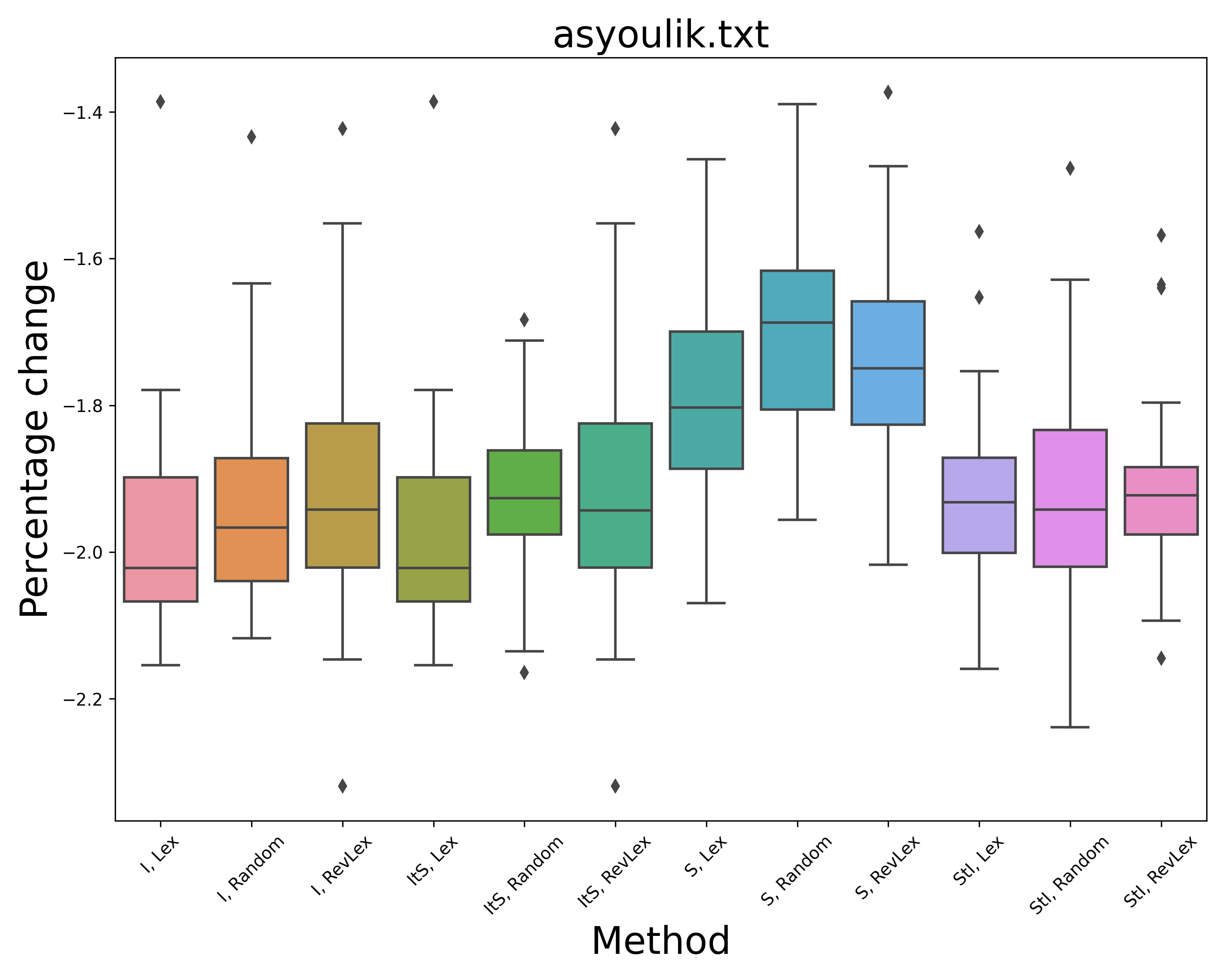}
    \caption{Percentage change in compression for \texttt{alice29.txt}, \texttt{fields.c}, and \texttt{asyoulik.txt} for each neighborhood method in local search. The variation in each box plot shows the difference made by distinct initializations. }
    \label{fig:combined-fitness-plots-local-search}
\end{figure}

\section{Conclusion and Future Work}
\label{sec:conclusion}

The BWT is an important string transformation, enabling a concise, searchable representation of a string. It relies on ordering the characters in the string and this is usually assumed to be ASCII ordering, without exploring whether alternatives might be more effective. 

We studied heuristics for the computationally hard RLBWT Alphabet Ordering Problem which takes a string \s{s} of length $n$ over an alphabet $\Sigma$ of size $\sigma$ and seeks an ordering of $\Sigma$ for $RLE(BWT(\s{s}, \Sigma))$ that minimizes $\norm{RLE(BWT(\s{s}, \Sigma))}$.
We have performed extensive benchmarking using files from the Canterbury corpus and implemented the experimentation on Super Computing Wales HPC.
We started by inspecting a large sample space of 240,000 randomly sampled alphabet orderings and found only limited improvement over the ASCII ordering.
This motivated searching local neighborhoods to improve fitness - this was achieved using a First-Improvement algorithm.

Various initializations have been applied to the test files to attempt to speed up the search which include: ASCII order, letter frequency and order of appearance, and a hand-tuned ordering given by \cite{chapin1998higher}.
Jumping around the complex landscape was implemented with neighborhood search using \textsc{Swap} and \textsc{Insert} operators as well as  combinations of these operators. 
Additionally varying the neighbors of an alphabet ordering has been explored, searching them lexicographically, reverse-lexicographically, and randomly. 
Overall, we inspected a combination of 9 initializations, 4 operators, and 3 neighborhood search methods, giving a total of 108 algorithm configurations. 

The number of search steps needed to outperform the best result of random sampling was found to be relatively few and could be computed in seconds, but quickly increased for achieving a local minimum. Indeed, we demonstrated that reaching a similar number of improvements with random sampling would require investigating a much larger number of random samples (e.g., $\gt10^{100}$), which is not feasible to compute.

The chosen initial alphabet order was found to influence the best solution that can be found, and we demonstrated this using the \textsc{Swap} operator. 
While we observed variation in the initial ordering which achieved the best fitness within a time limit, there is not necessarily one best initial ordering. However, all initializations exhibit a clear decrease in file size followed by a large number of steps until the local minimum is reached.

We observed that the random neighbor ordering methods perform well in the early stages of the search, and while not the best they remain competitive overall.
\textsc{Insert} and \textsc{Insert}-first methods are slower than \textsc{Swap} and \textsc{Swap}-first methods to reach a local minimum but will usually achieve a better compression.
Although our empirical evidence shows that local search is indeed effective for improving the RLBWT, and trade-offs occur, nonetheless we are still able to recommend a time-limited local search using
\textsc{Swap} with Random neighborhood exploration to improve rapidly upon the ASCII ordering. 
If more computational time is available to explore a better ordering then we recommend a local search using \textsc{Insert then Swap} with Lexicographic neighborhood exploration. 

We found it surprising that the ASCII ordering performs very well when compared to a random order. Also, it works relatively well as an initialization for the local search. 
If limited time is available we recommend the Chapin-Tate ordering to start the search, or a random order to initialize a longer search. 
However, there is no clear best initialization suiting different files considered in the corpus.

This is a difficult problem in theory and we have now demonstrated that this is a challenging problem in practice for the range of files in the Canterbury Corpus, with no clear winning strategy.
We have demonstrated that navigating trade-offs can be worthwhile for enhancing compressibility.
Our local search performs much better (faster convergence, better fitness) than random sampling and is useful even when computational time is limited.

In future work we intend to investigate further what constitutes a good alphabet ordering, the effect that different changes to an ordering can have on the transformed string, what factors contribute to the quality of an ordering for a given string and why some orderings perform better than others, in relation to the type of data (for instance natural language versus other structured data).

We want to determine which characters can be moved to benefit the search and to use this knowledge to inspire new and more specific local search operators. This may include using operators which re-order multiple characters at a time or incorporation of various crossover operators.
In addition different encoding methods for RLE may give better results and should be investigated.

\section{Statements and Declarations}

\subsection{Funding}

This work is supported by  the UKRI AIMLAC CDT,  \url{http://cdt-aimlac.org}, grant no. EP/S023992/1,
and was part-funded by the European Regional Development Fund through the Welsh Government, grant 80761-AU-137 (West).

\begin{centering}
\bigskip
\includegraphics[width=0.3\textwidth] {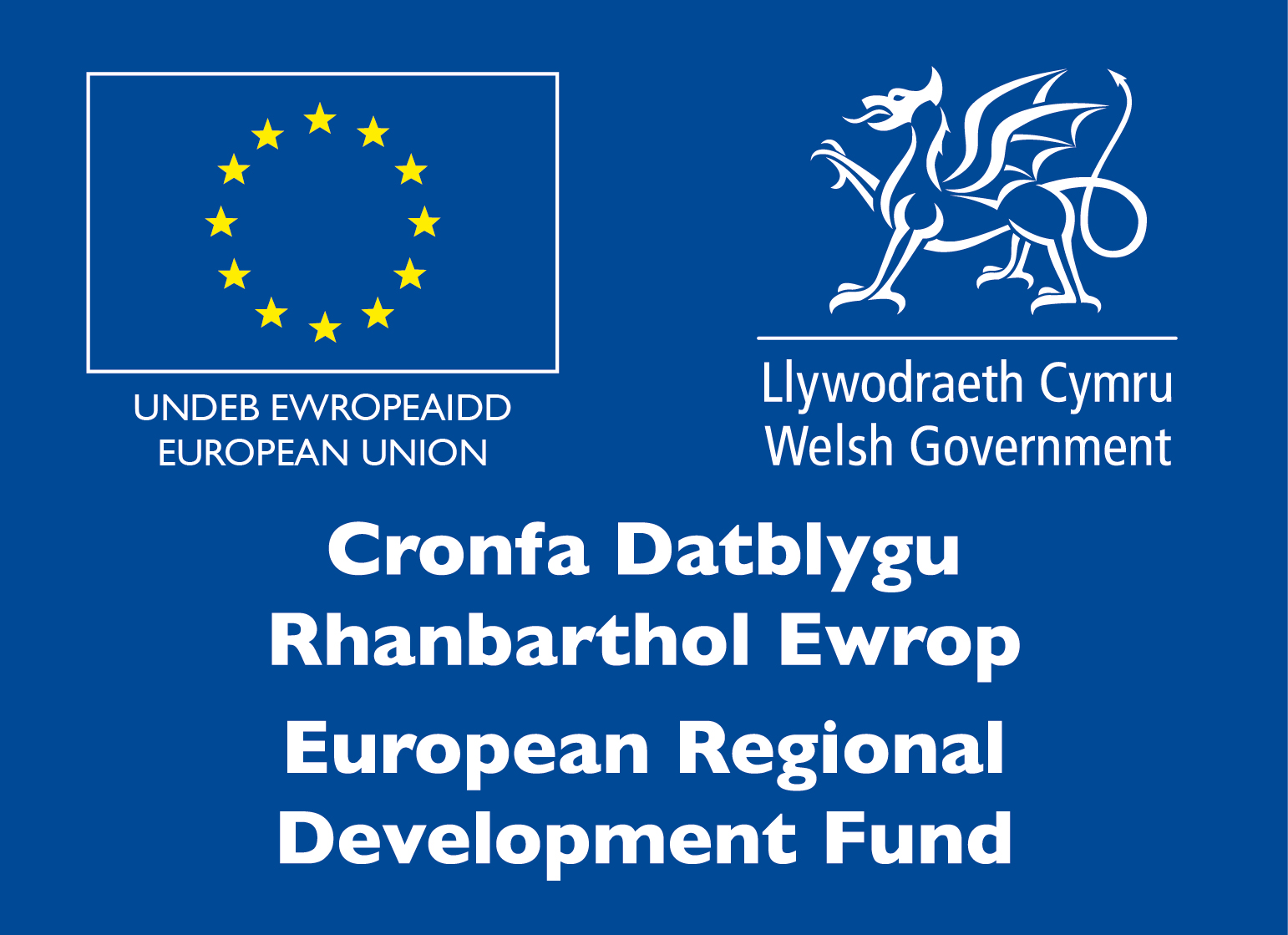}
\par
\end{centering}

\subsection{Acknowledgments}

We acknowledge the support of the Supercomputing Wales project, which is part-funded by the European Regional Development Fund (ERDF) via Welsh Government.

\subsection{Conflict of interest}

The authors declare that they have no conflict of interest or competing interests.

\subsection{Availability of code and data}

Code is available in our repository \url{https://github.com/jam86/Heuristics-for-the-Run-length-Encoded-Burrows-Wheeler-Transform-Alphabet-Ordering-Problem}, and \url{https://doi.org/10.5281/zenodo.8139504}.

Data from our experiments is available at \url{https://doi.org/10.5281/zenodo.8139367}.

\subsection{Author contributions}

All authors contributed to the study conception and design. Material preparation, data collection and analysis were performed by Lily Major. All authors wrote, read, and approved the final manuscript.

\clearpage

\bibliography{main}

\end{document}